\documentclass{article}

\usepackage{algorithm}
\usepackage{algorithmicx}
\usepackage[noend]{algpseudocode}
\usepackage{amsfonts}
\usepackage{amsmath} 
\usepackage{amssymb}
\usepackage{amsthm}
\usepackage{authblk}
\usepackage[english]{babel}
\usepackage{balance}  
\usepackage{booktabs} 
\usepackage{enumerate} 
\usepackage{epsfig}
\usepackage{epstopdf}
\usepackage{float}
\usepackage[T1]{fontenc}
\usepackage{graphicx}
\usepackage{hyperref}
\usepackage[utf8]{inputenc}
\usepackage{paralist}
\usepackage{stfloats}
\usepackage{subcaption}
\usepackage{tikz}
\usepackage{xspace}
\usepackage{mathtools}
\usepackage{geometry}
\usepackage{todonotes}

\usetikzlibrary{shapes}

\newcommand{\kmeans}{{$k$-\textsc{means}}\xspace}
\newcommand{\kmed}{{$k$-\textsc{median}}\xspace}
\newcommand{\kcent}{{$k$-\textsc{center}}\xspace}

\newcommand{\sat}{\textsc{SAT}\xspace}
\newcommand{\qsat}{\textsc{QSAT}\xspace}
\newcommand{\qsatb}{\textsc{QSAT-B}\xspace}
\newcommand{\tsat}{\textsc{3SAT}\xspace}
\newcommand{\tsatb}{\textsc{3SAT-B}\xspace}

\newcommand{\uqsat}{\textsc{U-QSAT}\xspace}
\newcommand{\uqsatfull}{\textsc{Unambiguous QSAT}\xspace}
\newcommand{\uqsatb}{\textsc{U-QSAT-B}\xspace}

\newcommand{\utdmb}{\textsc{U-3DM-B}\xspace}
\newcommand{\utdmbfull}{\textsc{Unambiguous 3DM-B}\xspace}
\newcommand{\ucbtb}{\textsc{U-CBT-B}\xspace}
\newcommand{\ucbtbfull}{\textsc{Unambiguous Covering By Triples-B}\xspace}

\newcommand{\stsatb}{\textsc{S-3SAT-B}\xspace}
\newcommand{\stqsatb}{\textsc{S-QSAT-B}\xspace}

\newcommand{\stsatbfull}{\textsc{Stable 3SAT-B}\xspace}

\newcommand{\stdmb}{\textsc{S-3DM-B}\xspace}
\newcommand{\stdmbfull}{\textsc{Stable 3DM-B}\xspace}
\newcommand{\scbtb}{\textsc{S-CBT-B}\xspace}
\newcommand{\scbtbfull}{\textsc{Stable Covering by Triples-B}\xspace}

\newcommand{\abs}[1]{\ensuremath{|#1|}}

\newcommand{\HW}[2]{\ensuremath{HW\left(#1, #2\right)}}

\newcommand{\bR}{{\mathbb R}\xspace}

\newcommand{\cC}{\ensuremath{\mathcal{C}}}

\newcommand{\cI}{\ensuremath{\mathcal{I}}}

\newcommand{\cT}{\ensuremath{\mathcal{T}}}
\newcommand{\cX}{\ensuremath{\mathcal{X}}}
\newcommand{\RR}{\mathbb{R}}

\newtheorem{theorem}{Theorem}
\newtheorem{lemma}{Lemma}
\newtheorem{claim}{Claim}
\newtheorem{corollary}{Corollary}
\newtheorem{definition}{Definition}

\newcommand{\cl}{{\mathcal X}\xspace}

\newcommand{\cost}{{\rm cost}\xspace}

\newcommand{\fa}{{\mathcal C}\xspace}
\newcommand{\opt}{{\mathcal O}\xspace}
\newcommand{\loc}{{\mathcal S}\xspace}

\newcommand{\eps}{{\epsilon}}

\title{Exact Algorithms and Lower Bounds for Stable Instances of Euclidean \kmeans\footnote{A preliminary version of this appeared in Proceedings of ACM-SIAM Symposium on Discrete Algorithms (SODA) 2019.}}
\author[1]{Zachary Friggstad\thanks{This research was undertaken, in part, thanks to funding from the Canada Research Chairs program and an NSERC Discovery Grant.}}
\author[2]{Kamyar Khodamoradi}
\author[1]{Mohammad R. Salavatipour\thanks{Supported by NSERC.}}
\affil[1]{Department of Computing Science\\ University of Alberta}
\affil[2]{Department of Computer Science\\ University of Regina}
\date{}

\begin{document}

\maketitle



\begin{abstract}
We investigate the complexity of solving stable or perturbation-resilient
instances of \kmeans and \kmed clustering in fixed dimension Euclidean metrics (or more generally doubling metrics).
The notion of stable or perturbation resilient instances was introduced by Bilu and Linial [2010] and 
Awasthi, Blum, and Sheffet [2012]. In our context, we say a \kmeans instance is $\alpha$-stable if
there is a unique optimum solution which remains unchanged if distances are (non-uniformly)
stretched by a factor of at most $\alpha$.
Stable clustering instances have been studied to explain why heuristics such as
Lloyd's algorithm perform well in practice. 
In this work we show that for any fixed $\epsilon>0$, $(1+\epsilon)$-stable instances of \kmeans in doubling metrics, which include fixed-dimensional Euclidean metrics, can be solved in polynomial time.
More precisely, we show a natural multi-swap local-search algorithm in fact finds the optimum solution for $(1+\epsilon)$-stable instances of \kmeans and \kmed in a polynomial number of iterations.

We complement this result by showing that this is essentially tight: that when the dimension $d$ is part of the input, there is a fixed $\epsilon_0>0$ such there is not even a PTAS for $(1+\epsilon_0)$-stable \kmeans in $\mathbb R^d$ unless NP=RP. To do this, we consider a robust
property of CSPs; call an instance stable if there is a unique optimum solution $x^*$ and for any other solution $x'$, the number of unsatisfied clauses is proportional to the Hamming distance between $x^*$ and $x'$. Dinur, Goldreich, and Gur have already shown stable QSAT is hard to approximation for some constant $Q$ \cite{DGG}. Recently, Paradise \cite{Orr} extended this to the setting with bounded variable occurrence. More specifically, it implies that stable QSAT with bounded variable occurrence is APX-hard.
Given this, we consider ``stability-preserving'' reductions to prove our hardness for stable \kmeans. Such reductions seem to be more fragile and intricate than standard $L$-reductions and may be of further use to demonstrate other stable optimization problems are hard to solve.
\end{abstract}
\thispagestyle{empty}

\newpage
\setcounter{page}{1}
\section{Introduction}
\label{sec:intro}

The interest in explaining the difference between performance of many heuristic algorithms (in particular for clustering
problems) in practice vs. worst-case performance bounds has recently attracted attention and led to new research directions.
It has been long observed that for many optimization problems, such as clustering problems, the performance of some well known
heuristics are much better than their worst case performance analysis.

There have been several approaches to study and explain these differences. 
Bilu and Linial \cite{BL12} as well as Awasthi, Blum, and Sheffet  \cite{ABS10} introduced the notion of stability and perturbation resilience. 
The idea is that for many problems, such as a clustering problem, a typical instance of the problem is stable in the sense that 
the optimum solution is unique and it does not change even if one modifies or perturbs input  parameters by a small factor.
Informally, instances of a problem are called $\alpha$-stable or $\alpha$-perturbation resilient
if the structure of the optimum solution remains unchanged even if the input is perturbed by an $\alpha$ factor.
For example, a clustering problem is $\alpha$-stable if there is a unique optimum solution which remains the unique optimum
after some distances are scaled up to a factor of $\alpha$: different pairs of points may have their distances scaled differently.

Balcan et al. \cite{BBG09} argue that for clustering problems the goal is to find
the ``target'' clustering and typically the objective function is just a proxy. Therefore, the distances of the input points and
how they  contribute to objective function are typically not very precise; thus small changes in these values usually does not
change the target clustering.
It has been shown that for $\alpha$-stable instances of several problems
such as centre-based clustering problems (e.g. \kcent, \kmed, \kmeans), graph partitioning problems  (e.g.
Max-cut, Multiway cut), and other problems, one can find the optimum solution in polynomial time. 

In this paper, we focus on $\alpha$-stable instances of the classical clustering problems \kmed and \kmeans
in Euclidean metrics $\RR^d$. Perhaps the most widely used clustering model 
is the \kmeans clustering: Given a set $\cX$ of $n$ {\em data points}
in $d$-dimensional Euclidean space $\RR^d$, and an integer $k$, 
find a set of $k$ points $c_1,\ldots,c_k \in \RR^d$ to act as as {\em centres}
that minimize the sum of squared distances of each data point to its nearest 
centre. In other words, we would like to partition $\cX$ into $k$ cluster sets, $\{C_1,\ldots,C_k\}$ and 
find a centre $c_i$ for each $C_i$ to minimize
\[\sum_{i=1}^k\sum_{x\in C_i}||x-c_i||^2_2. \]
Here, $||x-c_i||_2$ is the standard Euclidean distance in $\RR^d$ between points $x$ and $c_i$.
This value is called the cost of the clustering.
Typically, the centres $c_i$ are selected to be the centroid (mean) of the cluster $C_i$.
In other situations the centres must be from the data points themselves (i.e. $c_i\in C_i$) or from
a given set $\fa$. This latter version is referred to as discrete \kmeans clustering.
There are results that show that one can reduce \kmeans to discrete version at a small loss
(see \cite{Matousek00}).
In this paper we study discrete \kmeans.
The problem is known to be NP-hard even for $k=2$ for $\RR^d$ when $d$ is not
fixed or for arbitrary $k$ in $\RR^2$ \cite{ADHP09,DFKVV04,MNV09,Vattani09}. 
Several approximation algorithms have been proposed for the
problem; for a while the best ratio being a $(9+\epsilon)$ via a local search algorithm \cite{KMNPSW04}.
This was recently improved to a PTAS independently by \cite{FRS16A,FRS16} and \cite{CKM16} when the dimension $d$ can be regarded
as a constant and an 6.357-approximation for arbitrary dimensions \cite{ANSW17}. Awasthi et al. \cite{ACKS15} showed that \kmeans is APX-hard in $\RR^d$ when $d = \Omega(\log n)$. 
Improvements and extensions of this lower bound in high-dimensional spaces have appeared in \cite{LSW17,CSL19} with the current-best lower bound for discrete \kmeans in Euclidean spaces being 1.73 \cite{CSL22}. We briefly remark here that their techniques do not imply hardness of stable clustering problems, this will be discussed further in Section \ref{sec:results}.

We now precisely define what it means for an instance of \kmeans be stable in our paper.
One can similarly define what it means for a \kmed instance to be $\alpha$-stable.
We present a \kmeans instance
as a triple $(\cX, \fa, \delta)$ where $\delta$ is a symmetric distance function between points in $\cX \cup \fa$ that satisfies the triangle inequality unless
otherwise explicitly stated. A solution is viewed as s set $\loc \subseteq \fa$
with $|\loc| = k$ and its cost is $\cost(\loc) := \sum_{j \in \cX} \delta(j, \loc)^2$ where we let $\delta(j, \loc)$ to denote $\min_{i \in \loc} \delta(j, i)$.
\begin{definition}[$\alpha$-stability]
For $\alpha \geq 1$, call an instance $\cI = (\cX, \fa, \delta)$ of metric \kmeans $\alpha$-stable if it has a unique optimum
solution $\opt$ which is also the unique optimum solution in every related (not necessarily metric) instance
$\cI'=(\cX,\fa,\delta')$ (that need not satisfy the triangle inequality but still satisfies symmetry) with $\delta(i,j)\leq \delta'(i,j)\leq \alpha \cdot \delta(i,j)$ for all $i,j \in \cX \cup \fa$.
\end{definition}

Several papers have studied complexity of $\alpha$-stable instances of \kmeans and \kmed. The main goal is to find
algorithms that work for smaller values of $\alpha$ (i.e. weak requirement for stability).
Awasthi, Blum, and Sheffet \cite{ABS12} showed that 3-stable instances of \kmeans and \kmed can be solved in polynomial time.
Balcan and Liang \cite{BL16} improved this by showing that for $\alpha=1+\sqrt{2}$, $\alpha$-stable 
instances of \kmeans and \kmed
can be solved in polynomial time. This was further improved in the case of metric stability recently by Agelidakis, Makarychev, Makarychev \cite{AMM17} who
showed that 2-metric stable (or 2-metric perturbation resilient) instances of centre-based clustering problems
such as \kmeans and \kmed can be solved in polynomial time.

In this work we focus on discrete \kmeans and \kmed on Euclidean metrics $\RR^d$ (and more generally doubling metrics)
 and prove both upper and lower bounds.
We prove that for any fixed $\epsilon>0$, $(1+\epsilon)$-stable instances of these problems on fixed dimension Euclidean
spaces ($\RR^d$ for fixed $d$) can be solved in polynomial time and this is tight.
In fact, our result is slightly stronger in that we show $(1+\epsilon)$-stable instances do not even admit a PTAS unless NP=RP.

More specifically, starting with a new proximity PCP theorem explained below in Theorem \ref{PCP-hypothesis}
(which is a bounded-occurrence version of Theorem 3.1 of \cite{DGG}) we can show that for some fixed $\epsilon>0$,
$(1+\epsilon)$-stable instances of \kmeans and \kmed cannot be solved in polynomial time when restricted to $\RR^d$ 
(but unbounded $d$) unless NP=RP. We had presented Theorem \ref{PCP-hypothesis} as a hypothesis in
the conference version of our paper. Since then, Paradise \cite{Orr}
has proved a bounded-occurrence version of Theorem 3.1 of \cite{DGG} which directly implies Theorem \ref{PCP-hypothesis}.

\subsection{Previous work}
Bilu and Linial \cite{BL12} gave a polynomial exact algorithm for $O(n)$-stable instances
of Max-Cut. This was improved to $O(\sqrt{n})$-stable instances by Bilu, Daniely, Linial, Saks \cite{BDLS13} and further by 
Makarychev, Makarychev, Vijayaraghavan \cite{MMV14} who provided a polynomial exact algorithm based on semidefinite programming  
for $O(\sqrt{\log n}\log\log n)$-stable instances of Max-Cut. This result may be nearly tight: \cite{MMV14} also shows that solving $\alpha(n)$-stable instances in polynomial time would imply an $\alpha(2n)$-approximation for the nonuniform sparsest cut problem.

Awasthi, Blum, and Sheffet \cite{ABS12} showed that 
for 3-stable instances of large class of clustering problems, called separable centre-based objective (s.c.b.o) clustering problems (such as \kmed over finite metrics (with no Steiner points)) and for $(2+\sqrt{3})$-stable 
instances of s.c.b.o clustering problems with Steiner points one can find the optimum 
clustering in polynomial time. Furthermore, they proved NP-hardness for instances of \kmed with Steiner points that satisfy \emph{3-centre proximity} condition ($\alpha$-centre proximity is the condition that for any point $x \in \cX$ in cluster $C_i$ with cluster centre $c_i$,  $\alpha \cdot \delta(x, c_i) < \delta(x, c_j)$ if $i \neq j$).  
Ben-David and Reyzin \cite{BR14} showed the NP-hardness of $(2 - \epsilon)$-stable instances of \kmed with no Steiner points for general metrics. 
Balcan, Haghtalab, and White \cite{BHW16} prove that 2-stable instances of $k$-centre can be solved in polynomial time and any
$(2-\epsilon)$-stable instances of the problem are NP-hard.
Angelidakis, Makarychev, and Makarychev \cite{AMM17} show that for class of clustering problems called centre-based clustering
(such as \kmeans, \kmed, \kcent), 2-metric perturbation resilient instances can be solved in polynomial time,
improving the bound of $1+\sqrt{2}$ from Balcan and Liang \cite{BL16}.
Ostrovsky et al. \cite{ORSS12} showed that for $\epsilon$-separated (defined below) instances of continuous \kmeans a variant of Lloyd's algorithm
is an $O(1)$-approximation.

For an instance of continuous \kmeans with $\cX \in \mathbb R^d$, let $\Delta^2_k(\cX)$ denote the optimal \kmeans clustering cost.
Say that the instance given by $\cX$ is $\epsilon$-separated if $\Delta^2_k (\cX)\leq \epsilon^2 \Delta^2_{k-1} (\cX)$.
Ostrovsky et al. \cite{ORSS12} showed that one can achieve a $(1+f(\epsilon))$-approximation to \kmeans in polynomial time.
This result was further improved by Awasthi, Blum, and Sheffet \cite{ABS10} that if $\Delta^2_k (\cX)\leq \alpha \Delta^2_{k-1} (\cX)$ for
some constant $\alpha>1$ then one can obtain a PTAS for \kmeans in time polynomial in $n,k$ but exponential 
in $\alpha,\epsilon$.
A solution $S_0$ to \kmed is $\frac{1}{\epsilon}$-locally optimal if any solution $S_1$ such that
$|S_0-S_1|+|S_1-S_0|\leq 2/\epsilon$ has cost at least as big as that of $S_0$. Cohen-Addad and Schwiegelshohn \cite{CS17}
showed that for $\alpha>3$, for any instance of \kmed that is $\alpha$-stable, any $\frac{2}{\alpha-3}$-locally optimal
solution is optimum. Hence a local search algorithm that swaps up to $\frac{2}{\alpha-3}$ centres finds the optimum solution.
However, they do not show how to find such a local optima in polynomial time.

Vijayaraghavan, Dutta, and Wang \cite{VDW17} studied additive perturbation stable (APS) instances of Euclidean \kmeans for $k=2$. 
An instance is $\delta$-APS if the (unique) optimum clustering remains optimum even if each point is moved up to $\delta$.
They \cite{VDW17} showed that for any fixed $\epsilon>0$, $\epsilon$-additive instances of Euclidean \kmeans 
for $k=2$ can be solved in polynomial time. There are also several results on stable instances of graph partitioning problems such as Max-Cut.

Another interesting aspect of our result is that we prove the local search dynamics find the optimum solution in polynomial time in stable instances
of \kmed and \kmeans. This stands in stark contrast to the fact that the complexity of finding a local minimum for the standard local search algorithm
is PLS-complete for \kmed \cite{AKP08}. We refer the interested reader to \cite{AKP08} and the references therein for more details of Polynomial Local Search
complexity. Essentially, our result shows that if stable instances of \kmed and \kmed too are PLS-hard then the fact that local search terminates in a polynomial
number of iterations would allow us to solve all problems in PLS in polynomial time.
The fact that local search terminates in polynomial time for stable instances is not surprising, but we finally provide the first proof of this fact.

\subsection{Our Results}\label{sec:results}
Our main results are Theorems \ref{thm:alg} and \ref{thm:hardness} below.
Recall a metric has doubling dimension $d$ if any ball with radius $R$ in the metric can be covered by at most $2^d$ balls of radius $R/2$.
Thus, $d$-dimensional Euclidean metrics have doubling dimension $O(d)$.

\begin{theorem}\label{thm:alg}
For any fixed $d \geq 1$ and $\epsilon'>0$, $(1+\epsilon')$-stable instances of \kmeans and \kmed in metrics with doubling dimension $d$ can
be solved in polynomial time.
\end{theorem}

This theorem is proved by showing that the simple $\rho$-swap local search algorithm for a suitable constant $\rho=\rho(\epsilon',d)$
finds the optimum \kmeans clustering in polynomial time if the best improvement is taken in each iteration.
We should note that, in all the previous studies of local search algorithms, 
in order to obtain polynomial run time, a swap is performed if it yields a ``significant'' improvement of the solution.
Hence, the result of the algorithms is not a true local optimum, but, in some sense, an approximate local optimum.
However, in order to find the actual optimum, one cannot rely on an algorithm that produces an approximate local optimum.
For instance, the result of Cohen-Addad and Schwiegelshohn \cite{CS17} shows that a {\em true} local optimum is also optimum
in $\alpha$-stable instance of \kmed (for $\alpha>3$) but it does not show how to find a true local optimum in polynomial
time. In order to prove Theorem \ref{thm:alg} we must show that the local search algorithm that performs the best swap in each step
in fact finds the true local (and hence global) optimum in stable instances. We focus only on our setting of doubling metrics, but the ideas
can also be used to show how to find the global optimum in polynomial time for $\alpha$-stable instances of \kmeans and \kmed in general metrics studied in \cite{CS17} for constant $\alpha > 3$.

As a side effect, we also show how to avoid the ``$\eps$-loss'' that so many local search procedures lose when being modified to run in polynomial time.
For example, a local optimum solution to the single-swap heuristic for \kmed is known have cost at most $5$ times the global optimum cost, yet
a modification to the standard single-swap algorithm in \cite{ARYA} to ensure the algorithm runs in polynomial time is a $(5+\epsilon)$-approximation.
We provide analysis of the local-search procedure that takes the best swap at each step and prove the solution is a true 5-approximate solution (no $\eps$-loss)
after a polynomial number of iterations, even if it has not yet stabilized at a local optimum. Our approach may be helpful for others to communicate more clean approximation ratios
for their local search algorithms. The details of this analysis technique appear in Appendix \ref{app:alt}.

Our second major result is to show that \autoref{thm:alg} is essentially tight in that the assumption of $d$ being constant is critical to allow instances with arbitrarily small (constant) stability to be solved optimally.

Part of our reduction is inspired by recent work on \kmeans hardness in Euclidean spaces, notably on ideas in the reduction from \cite{ACKS15}. It should be noted that their reduction has each point being at a distance $1$ or some constant $c > 1$ from every possible centre and that in the ``yes'' case there is a solution where every point is within distance $1$ of a centre. But this does not guarantee the instance is stable, even if one follows parsimonious reductions from \uqsatfull (definition below) to ensure there is a unique solution where each point pays $1$ in its cluster. The simple reason is that stability is a much stronger requirement of a problem: slight perturbation of distances could change the optimal solution structure even in instances with distances $1$ and $c > 1$. It is not difficult to come up with such examples.

In order to prove \kmeans is hard even in stable instances, we
prove \autoref{thm:hardness} using a new PCP construction by Paradise \cite{Orr}, which is a slightly stronger version
of Theorem 3.1 in \cite{DGG}. We introduced this PCP construction as a hypothesis in the conference version of this paper. Since then, Paradise \cite{Orr} has actually proved
this hypothesis. We need the following definitions to state our result formally.

\begin{definition}[\uqsatfull]
In an instance of promise problem \uqsat, we are given a set of $n$ variables $x_1, \, x_2, \, \ldots, \, x_n$ and $m$ clauses $C_1, \, C_2, \, \ldots, \, C_m$ where each clause is a $Q$-CNF over $Q$ (distinct) variables.
The promise is that there is at most one satisfying assignment.
\end{definition}

\uqsat was proven hard using a randomized reduction in \cite{VV86} in the sense that an algorithm that can be used to determine satisfiability of a \uqsat instance could then be used to solve any language in $NP$ with a randomized reduction. That is, we would have $NP=RP$.

Given two binary vectors $x,x'$ of the same length, let $HW(x,x') \in [0,1]$ denote the Hamming weight of $x,x'$: the fraction of coordinates $i$ with $x_i \neq x'_i$.
The following is a corollary of Theorem 3.1 of \cite{DGG}, obtained by producing their PCPP for a given instance of \uqsat.
\begin{theorem}(Theorem 3.1 of \cite{DGG})\label{DGG-theorem}
There are universal constants $q,s,\epsilon>0$ such that the following holds.
For every $L\in NP$ there is a polynomial time randomized reduction from $L$ to an instance $\Phi$ of \uqsat with the following properties:

\begin{description}
\item{\bf Yes case:} if $L$ is a yes case then $\Phi$ has a unique satisfying assignment $x^*$ with probability $\Omega(1/{\rm poly}(n))$. Also, for any assignment $x$ to $\Phi$,
the fractions of clauses not satisfied by $x$ is at least $s \cdot HW(x,x^*)$.

\vspace{-1.5mm}
\item{\bf No case:} if $L$ is a no case then no assignment satisfies more than $(1-\epsilon)$-fraction of clauses of $\Phi$.
\end{description}
\end{theorem}

An instance of \uqsatb is the same as \uqsat with the additional condition that
each variable appears in a bounded number of clauses. We typically use $B$ to refer to this bound as well. The following bounded occurrence version of Theorem \ref{DGG-theorem} is the basis for our hardness result.
It simply repeats Theorem \ref{DGG-theorem} with the condition that the \sat instance has bounded occurrence for each variable.

\begin{theorem}\label{PCP-hypothesis}[Hypothesis 1 in \cite{FKS19}]
There are universal constants $B,Q,s,\epsilon>0$ such that:
For every $L\in NP$ there is a polynomial time (randomized) reduction from $L$ to an instance $\Phi$ of \uqsatb with the following properties:
\vspace{-1.5mm}
\begin{description}
\item{\bf Yes case:} if $L$ is a yes case then $\Phi$ has a unique satisfying assignment $x^*$ with probability $\Omega(1/{\rm poly}(n))$. Also, for any assignment $x$ to $\Phi$,
the fractions of clauses not satisfied by $x$ is at least $s \cdot HW(x,x^*)$.

\vspace{-1.5mm}
\item{\bf No case:} if $L$ is a no case then no assignment satisfies more than $(1-\epsilon)$-fraction of clauses of $\Phi$.
\end{description}
\end{theorem}

As mentioned above, we stated the above as a hypothesis in the conference version of this paper \cite{FKS19}. After the conference appeared, Paradise \cite{Orr} proved a stronger version of 
Theorem 3.1 of \cite{DGG}. Rather than recalling all definitions in \cite{Orr}, we simply summarize the main features of their result. Theorem 1.6 in \cite{Orr} shows every $L \in NP$ admits a PCP verifier that is:
\begin{itemize}
\item {\bf Strong}: The reduction to the PCP is parsimonious in that there is a 1-to-1 correspondence between witnesses for yes instances of $L$ and proofs that are accepted with probability exactly 1 by the verifier: the so-called {\bf canonical proofs} $\mathcal C$.
Further, for any proof $\pi$ the probability the PCP rejects $\pi$ is $\Omega(HW(\pi, \mathcal C))$: i.e. at least proportional to the minimum Hamming distance between $\pi$ and a proof in $\mathcal C$.
\item {\bf Smooth}: All bits of the proof are queried by the verifier with the equal probability.
\end{itemize}

The proof of Theorem \ref{PCP-hypothesis} is then immediate. First, use the reduction in \cite{VV86} to reduce $L$ to an instance of the promise problem \uqsat. Then apply the PCP verifier from Theorem 1.6 in \cite{Orr} followed by the reduction to \qsatb from the proof of Corollary 1.8
in \cite{Orr}. This directly establishes Theorem \ref{PCP-hypothesis} with no further arguments being required.
Note that the randomization in the reduction only comes from the reduction in \cite{VV86} proving \uqsat is hard, every step from \cite{Orr} is deterministic.


We will frequently refer to SAT instances with the properties mentioned in the Yes case of Theorem \ref{PCP-hypothesis}, so it will be convenient to use two more definitions.
\begin{definition}[Stable SAT Instances]
For $0 \leq s \leq 1$,
an instance $\Phi$ of $\sat$ is said to be $s$-stable if there is exactly one satisfying assignment $x^*$ and for any assignment $x$ the fraction of clauses of $\Phi$ that are not satisfied is at least $s \cdot \HW{x^*}{x}$.
\end{definition}
\begin{definition}[\stqsatb]
In an instance of promise problem \stqsatb, we are given an instance $\Phi$ of \qsatb with the following guarantee: 
Either $\Phi$ is $s$-stable or it is not satisfiable.
\end{definition}
So, Theorem \ref{PCP-hypothesis} gives a randomized reduction from any $L \in NP$ to \qsatb that always maps a no instance to a no instance and,
with polynomially-large probability, maps a yes instance to an instance of \stqsatb which, by definition, has exactly one satisfying assignment.

Using Theorem \ref{PCP-hypothesis} we prove our main result for hardness of approximating Stable $k$-Means in Euclidean metrics whose dimension is not fixed.

\begin{theorem}\label{thm:hardness}
There exists universal constants $\epsilon, \gamma > 0$ such that there is no $(1+\gamma)$-approximation for
$(1+\epsilon)$-stable instances of \kmeans in $\RR^d$ unless $NP=RP$.
\end{theorem}
Before our work, it was not known that it was even hard to solve stable \kmeans in Euclidean spaces exactly much less hard to approximate.

~

{\bf Hardness Reduction Techniques:}
Our goal is to take a stable instance of \qsatb and map it to a stable instance of \kmeans. The definitions of stability in these two problems, of course, differ and our goal
is to not only provide a hardness reduction for \kmeans but also to translate the notion of stability from \qsatb to \kmeans. To this end, we informally call a reduction ``stability-preserving''
if it maps stable instances of one problem to stable instances of the other problem.

A strong caution to the reader is that standard $L$-reductions do not always preserve stability.
In decision problems like \qsatb or \textsc{3D-Matching} (an auxiliary problem we encounter on the way to proving hardness for stable \kmeans),
we certainly need our reduction to be parsimonious to ensure uniqueness of the optimal solution. But even parsimonious $L$-reductions do not suffice. In Appendix \ref{app:nostable},
we given an example of a stable \qsat and show the classic parsimonious $L$-reduction from \qsat to \qsatb for some constant $B$ that is based on expander graphs fails to preserve stability\footnote{Paradise \cite{Orr} shows this reduction does preserve stability if the \qsat instance has each variable appearing in the same number of clauses, which
would follow from the smooth property of their PCP.}.
Ultimately, this shows why we required the newer construction by Paradise \cite{Orr} rather than the previous construction in \cite{DGG} providing a strong PCP for every language in NP but does not guarantee smoothness.

Thus, stability-preserving reductions are more fragile than $L$-reductions. The arguments about why stability is preserved are also more challenging and in-depth than the standard ``no-case'' analysis in an $L$-reduction.
We believe such reduction may be interesting in other contexts, especially in proving hardness for other problems under certain stability assumptions.

{\bf Outline of the paper:}
The algorithm for solving stable instances of \kmeans in constant-dimension doubling metrics
appears in Section \ref{sec:algorithm}. The presentation focuses only on \kmeans, the algorithm for solving stable instances of \kmed
is nearly identical.

Some details of the hardness reduction are given in Section \ref{sec:hardness0}.
Then the reduction is broken into three steps. In Section \ref{sec:hardness1}, we begin by reducing \stqsatb to \stsatb which also serves as our introduction to the concept of 
stability-preserving reductions. In Section \ref{sec:hardness2}, we provide a stability-preserving
reduction from \stsatb to the classic \textsc{3D-Matching} problem with appropriate stability and bounded-degree
guarantees maintained in the reduction. A further step in reduction is needed to reduce from stable instances of \textsc{3D-Matching} to a covering problem. This simple step appears in Section \ref{sec:hardness3}.
Finally, Section \ref{sec:hardness4} finishes with the reduction 
to an $\alpha$-stable Euclidean \kmeans instance in $\mathbb R^d$ with $d = \Omega(n)$ for some absolute constant $\alpha > 1$ that, ultimately, depends on the $Q, B$ and $s$ from Theorem \ref{PCP-hypothesis}.
The $\eps$ from Theorem \ref{PCP-hypothesis} also factors into showing hardness of even approximating $\alpha$-stable \kmeans instances in high-dimensional Euclidean metrics within some small constant factor.


\section{Solving Low-Dimensional Stable Instances} \label{sec:algorithm}
Our main goal in this section is to prove Theorem \ref{thm:alg}. We prove it for \kmeans only, the proof for \kmed is essentially identical.
Suppose $\eps',d$ are fixed constants and we are given an instance $(\cl,\fa,\delta)$ of \kmeans in a doubling metric with doubling dimension $d$
that is $(1+\eps')$-stable; i.e. it has a unique optimum solution
$\opt \subseteq \fa$ and it remains the unique optimum solution even if distances between points in $\cX \cup \fa$ are scaled (non-uniformly)
by at most $(1+\eps')$ factor.
If the reader is not comfortable with doubling metrics, nothing will be lost by thinking of $\RR^d$ in particular whose doubling dimension is $\Theta(d)$.


Let $\eps$ be such that $1+6\eps = (1+\eps')^2$, roughly speaking we have $\eps \approx \eps'/3$ for small $\eps'$. 
Without loss of generality, we assume $\eps \leq 1/6$ since we can shrink $\eps'$ if necessary;
an instance that is $(1+\eps')$-stable is also $(1+\eps'')$-stable for $\eps'' < \eps'$.
Let $\rho := \rho(\eps, d)$ be the constant from the \kmeans local search algorithm in \cite{FRS16A,FRS16}.
We briefly recall that the $\rho$-swap local search analysis for \kmeans in Euclidean metrics of dimension $d$
finds a solution whose cost is at most $1+\eps$ times the optimum solution cost. For impressions, $\rho = d^{O(d)} \cdot \eps^{-O(d/\eps)}$.

Let $\mathfrak{F}_k = \{\loc \subseteq \fa : |\loc| = k\}$ be set of {\em feasible} solutions.
Consider Algorithm \ref{alg:local}, which is the standard $\rho$-swap local search algorithm with slight modification that
in each step it performs the swap that yields the best improvement.

\begin{algorithm*}[h]
 \caption{$\rho$-Swap Local Search} \label{alg:local}
\begin{algorithmic}
\State let $\loc$ be any set in $\mathfrak{F}_k$
\While{$\exists$ sets $\loc' \in \mathfrak{F}_k$ with $|\loc - \loc'| \leq \rho$ and $\cost(\loc') < \cost(\loc)$}
\State $\loc \leftarrow \displaystyle \arg\min_{\substack{\loc' \in \mathfrak{F}_k\\|\loc-\loc'| \leq \rho}} \cost(\loc')$
\EndWhile
\State \Return $\loc$
\end{algorithmic}
\end{algorithm*}

Each iteration runs in polynomial time because $\rho$ is a constant. Unlike standard polynomial-time local search algorithms that stop once no improvement by a factor of $\eps/k$ can be made, our
algorithm simply terminates once no improvement is possible at all. We will argue that the algorithm terminates in a polynomial number of iterations in $(1+\eps')$-stable instances and
give an explicit bound on the number of iterations.

An interesting observation on the quality of the returned solutions in non-stable instances is made at the end of this section: informally it says
that if we truncate the main loop of Algorithm \ref{alg:local} to a polynomial number of iterations, then in $(1+\eps')$-stable instances it finds the optimum solution and, further, in arbitrary instances of \kmeans
in $\mathbb R^d$ it finds a $(1+O(\eps))$-approximate solution.

Let $\opt \in \mathfrak{F}_k$ be the unique optimum solution. For any set $\loc \in \mathfrak{F}_k$, define the following:
\begin{itemize}
\item For $j \in \cl$, let $\sigma(j, \loc)$ be the centre in $\loc$ nearest to to $j$, breaking ties by some fixed ordering of $\fa$.
\item $\overline{\cl}_\loc = \{j \in \cl : \sigma(j, \loc) \in \loc-\opt \text{ and } \sigma(j, \opt) \in \opt-\loc\}$.
\item $\Psi(\loc) = \sum_{j \in \overline{\cl}_\loc} \delta(j, \sigma(j, \loc))^2 + \delta(j, \sigma(j, \opt))^2$.
\end{itemize}

Here is why we define the function $\Psi(.)$. In the analysis of local search algorithms such as in \cite{FRS16A,FRS16}, in order
to show that a local optimum solution $\loc$ is $(1+\eps)$-approximate, one shows that 
$\cost(\loc)\leq \cost(\opt)+O(\eps)(\cost(\opt)+\cost(\loc))$. That bound is too crude for our purposes here.
Instead, we require $\cost(\loc) \leq \cost(\opt) + \eps \cdot \Psi(S)$, i.e. the error term is not an $\eps$ factor
of $\cost(\opt)+\cost(\loc)$; instead it is only an $\eps$ factor of the cost of $\opt$ and $\loc$ for points in $\overline{\cl}_\loc$.
The function $\Psi(\loc)$ is a bit challenging to track, it does not necessarily decrease as $\cost(\loc)$ decreases. Still, it is an important quantity in our analysis. 

\begin{definition}
Say $\loc \in \mathfrak{F}_k$ is a {\bf nearly-good solution} if $\cost(\loc) \leq \cost(\opt) + 2\eps \cdot \Psi(\loc)$.
\end{definition}


\subsection{A Structural Theorem}

We fix some $\loc \in \mathfrak{F}_k$ in this section, which may or may not be a local optimum solution.
Let $\loc' = \loc-\opt$ and let $\opt' = \opt-\loc$. We will use analysis from \cite{FRS16A,FRS16} to handle clients in $\overline{\cl}_\loc$.

For now, consider the \kmeans instance with points $\overline{\cl}_\loc$, possible centres $\fa' := \loc' \cup \opt'$, and $k' = |\loc'| = |\opt'|$. Note $\loc', \opt'$ are disjoint, which was a technical requirement for the analysis in \cite{FRS16A,FRS16}.
For a subset $P \subseteq \loc' \cup \opt'$ with $|P \cap \loc'| = |P \cap \opt'|$, define $\Delta_j^P$ for each $j \in \overline{\cl}_\loc$ to be $\delta(j, \sigma(\loc'  \triangle P))^2 - \delta(j, \sigma(\loc'))^2$, i.e.
the change in $j$'s assignment cost if we replaced solution $\loc'$ with $\loc' \triangle P$.

Now, it is not necessarily clear that $\opt'$ is an optimal solution for this restricted instance of \kmeans nor is it clear that $\loc'$ is a locally-optimal solution with respect to the $\rho$-swap heuristic. However, we do not need these facts. No inequality in \cite{FRS16A,FRS16} requires that $\opt'$ be a global optimum solution (it is just used to conclude that a local optimum solution is a near-optimal solution). Furthermore, all upper bounds on $\Delta_j^P$ terms in \cite{FRS16A,FRS16} do not require local optimality of $\loc'$: local optimality was used there simply to show $\Delta_j^P \geq 0$ if $|P| \leq \rho$ but we do not use this bound.

Thus, the following holds in our setting which follows from results in \cite{FRS16A,FRS16}.
\begin{theorem}\label{thm:structure}
There is a randomized algorithm that samples a partition $\pi$ of $\loc' \cup \opt'$ with the following properties:
\begin{itemize}
\item $|P \cap \loc'| = |P \cap \opt'| \leq \rho$ for each $P \in \pi$ and

\item 
${\bf E}_{\pi}\left[\sum_{P \in \pi} \sum_{j \in \overline{\cl}_{\loc}}\Delta_j^P]\right]\leq \sum_{j \in \overline{\cl}_{\loc}} (1+\epsilon) \cdot \delta(j, \opt')^2 - (1-\epsilon) \cdot \delta(j, \loc')^2$.
\end{itemize}
\end{theorem}
The randomized partition is described by Theorem 3.2 in \cite{FRS16} and the upper bound on the expected cost change when performing this swap is given at the end of Section 3.2 in \cite{FRS16} (for the appropriate choice of $\rho$). A careful reader can verify Theorem \ref{thm:structure} indeed holds simply by reading Section 3.2 in \cite{FRS16} plus appropriate definitions preceding that section. Instead of reiterating every argument from that paper to show this holds, we simply point out that Theorem \ref{thm:structure} does not require optimality of $\opt'$ nor local optimality of $\loc'$. Again for emphasis, in \cite{FRS16} local optimality was only used to show $0 \leq E_{\pi}\left[\sum_{P \in \pi} \Delta^P_j\right]$ which we do not require.

We use Theorem \ref{thm:structure} to show the following main technical result that supports our analysis of the running time of Algorithm \ref{alg:local} and proves it returns $\opt$ in stable instances is the following.
Intuitively, it says that if the solution is not nearly-good then the next step of the local search algorithm will make a significant step.

\begin{theorem}\label{thm:support}
For any $\loc \in \mathfrak{F}_k$, if $\cost(\loc) > \cost(\opt) + \eps \cdot \Psi(S)$  then
there is some $\loc' \in \mathfrak{F}_k$ with $|\loc - \loc'| \leq \rho$ where
\[ \cost(\loc') \leq \cost(\loc) + \frac{\cost(\opt) - \cost(\loc) + \eps \cdot \Psi(S)}{k}. \]
\end{theorem}

\begin{proof}
Sample a random partition $\pi$ of $\opt' \cup \loc'$ as in Theorem \ref{thm:structure} but now consider the effect of the 
swap $\loc \rightarrow \loc \triangle P$ for each part $P \in \pi$. We place an upper bound on
${\bf E}_\pi [\sum_{P \in \pi} \cost(\loc \triangle P) - \cost(\loc)]$ by describing a valid way to redirect each $j \in \cl$ 
in each swap on a case-by-case as follows. For brevity, let $c^*_j = \delta(j, \sigma(j, \opt))^2$ be the cost
of connecting $j$ in the global optimum solution and, analogously, $c_j = \delta(j, \sigma(j, \loc))^2$.

\begin{itemize}
\item We never move any $j$ with both $\sigma(j, \loc), \sigma(j, \opt) \in \loc \cap \opt$. Note $\sigma(j, \loc)$ remains open after each swap so this is valid. Observe for such clients that $c^*_j = c_j$ so we, conveniently,
say the total assignment cost change for $j$ over all swaps $P \in \pi$ is bounded by $c^*_j - c_j$.
\item For $j$ with $\sigma(j, \loc) \in \loc'$ and $\sigma(j, \opt) \in \loc \cap \opt$, move $j$ to $\sigma(j, \opt)$ when swapping the part $P$ with $\sigma(j, \loc) \in P$. 
As $\sigma(j, \loc)$ remains open when swapping all other $P' \neq P$, we can leave $j$ assigned to $\sigma(j, \loc)$ to upper bound its cost change for swaps $P' \neq P$ by 0. The total cost assignment for $j$ is then bounded by $c^*_j - c_j$.
\item For $j$ with $\sigma(j, \loc) \in \loc \cap \opt$ and $\sigma(j, \opt) \in \opt'$, move $j$ to $\sigma(j, \opt)$ when swapping the part $P$ with $\sigma(j, \opt) \in P$ and do not move $j$ when swapping any other part $P' \neq P$.
This places an upper bound of $c^*_j - c_j$ on the total assignment cost change for $j$.
\item Finally, consider $j$ with $\sigma(j, \loc) \in \loc'$ and $\sigma(j, \opt) \in \opt'$. Note these are precisely the points $j \in \overline{\cl}_S$. From Theorem \ref{thm:support},
\[{\bf E}_{\pi}\left[\sum_{P \in \pi} \Delta^P_j\right] \leq (1+\eps) \cdot c^*_j - (1-\eps) \cdot c_j = c^*_j - c_j + \eps \cdot (c_j + c^*_j).\]
\end{itemize}
Aggregating this cost bound for all clients, we see
\[ {\bf E}_{\pi} \left[\sum_{P \in \pi} \cost(\loc \triangle P) - \cost(\loc)\right] \leq \cost(\opt) - \cost(\loc) + \eps \cdot \Psi(S).\]
Therefore there is some $\pi$ and some $P \in \pi$ with
\[ \cost(\loc \triangle P) - \cost(\loc) \leq \frac{\cost(\opt) - \cost(\loc) + \eps \cdot \Psi(S)}{|\pi|} \leq \frac{\cost(\opt) - \cost(\loc) + \eps \cdot \Psi(S)}{k}, \]
the latter bound using $|\pi| \leq k$ and the fact the numerator is negative.
\end{proof}


\subsection{Polynomial-Time Convergence to a Nearly-Good Solution}
In order to show that Algorithm \ref{alg:local} terminates in polynomial time on stable instances, we first
show that a nearly-good solution will be encountered by Algorithm \ref{alg:local} within a polynomial number of iterations even if
the instance is not stable. 
The next section shows that the only nearly-good solution is the optimal solution in $(1+\eps')$-stable instances, thereby completing the proof of Theorem \ref{thm:alg}.

From Theorem \ref{thm:support}, we show solutions that are not nearly-good are improved significantly in a single step of the local search algorithm (i.e. in terms of their full cost, not just when considering clients in $\overline{\cl}_{\loc}$).
\begin{lemma}\label{lem:neargood}
Suppose $\loc \in \mathfrak{F}_k$ is a solution which is not nearly-good. Then there is some $\loc' \in \mathfrak{F}_k$ with $|\loc-\loc'| \leq \rho$
satisfying \[\cost(\loc') - \cost(\opt) \leq \left(1 - \frac{1}{2k}\right) \cdot (\cost(\loc) - \cost(\opt)).\]
\end{lemma}
\begin{proof}
Consider the set $\loc'$ guaranteed by Theorem \ref{thm:support}. The fact that $\loc$ is not a nearly-good solution means
\begin{eqnarray*}
\cost(\loc') - \cost(\opt) & \leq & \cost(\loc) - \cost(\opt) + \frac{\cost(\opt) - \cost(\loc) + \eps \cdot \Psi(S)}{k} \\
& < & \cost(\loc) - \cost(\opt) + \frac{\cost(\opt) - \cost(\loc)}{2k}
= \left(1 - \frac{1}{2k}\right) \cdot (\cost(\loc) - \cost(\opt)).\vspace{-1.5mm}
\end{eqnarray*}
\vspace{-2.5mm}
\end{proof}

To argue about the number of iterations of Algorithm \ref{alg:local},
we make the assumption that all coordinates of all points in $\cl \cup \fa$ are integers. This is without loss of generality:  scaling all points by the product of all denominators increases the bit complexity
of the input by a polynomial factor and Algorithm \ref{alg:local} would behave exactly as it would before the scaling (i.e. would consider the same sequence of sets $\loc$).

Let $\Delta = \max_{j \in \cl, i \in \fa} \delta^2(i, j)$. Observe that $\cost(\loc) - \cost(\opt) \leq n \Delta$, $\cost(\loc)$ is an integer for any $\loc \in \mathfrak{F}_K$, and $\ln \Delta$ is polynomial
in the bit complexity of the input.
\begin{corollary} \label{cor:neargood}
When Algorithm \ref{alg:local} terminates, the returned solution is a nearly-good solution. Also, within $2k \cdot \ln(n \Delta)$ iterations Algorithm \ref{alg:local} will have had some iteration with $\loc$
being a nearly-good solution.
\end{corollary}
\begin{proof}
Lemma \ref{lem:neargood} shows that if $\loc$ is not a nearly-good solution then there is a better solution $\loc' \in \mathfrak{F}_k$ with $|\loc-\loc'| \leq \rho$, so the local search algorithm can only terminate
with a nearly-good solution.

For the sake of contradiction, suppose that after $K = \lfloor2k \cdot \ln(n \Delta)\rfloor$ iterations  Algorithm \ref{alg:local} has still not encountered a nearly-good solution.
Say $\loc_0, \loc_1, \ldots, \loc_K \in \mathfrak{F}_k$ is the sequence of sets held by the algorithm after the first $K$ iterations, where $\loc_0$ is the initial set.

For $0 \leq i < K$, Lemma \ref{lem:neargood} and the fact that Algorithm \ref{alg:local} always chooses the best improving swap
shows $\cost(\loc_{i+1}) - \cost(\opt) \leq (1-1/(2k)) \cdot (\cost(\loc_i) - \cost(\opt))$. Thus,
\[ \cost(\loc_K) - \cost(\opt) \leq \left(1 - \frac{1}{2k}\right)^K \cdot (\cost(\loc_0) - \cost(\opt)) \leq \left(1 - \frac{1}{2k}\right)^K \cdot n \Delta < 1. \]
Because costs are integral, $\cost(\loc_K) - \cost(\opt) = 0$ which contradicts that $\loc_K$ is not a nearly-good solution.
\end{proof}

This does not yet show that Algorithm \ref{alg:local} terminates in a polynomial number of steps. This fact will be proven after the next subsection when we show the only nearly-good solution
in stable instances is $\opt$.


\subsection{Nearly-Good Solutions are Optimal in Stable Instances}
Our high-level approach is inspired by \cite{CS17}, but we must address larger technical challenges. Roughly speaking, the added difficulty is because the local search analysis
from \cite{FRS16A,FRS16} we are following has the bound on the cost change of the swaps depending mildly on $\cost(\loc)$. We are also burdened with proving that the optimum is found in a
polynomial number of iterations, something that was not addressed in \cite{CS17}.

Throughout this section, let $\loc$ be a fixed nearly-good solution. Define perturbed distances $\delta'(i,j)$ for $i \in \fa, j \in \cl$ as follows:
\[
\delta'(i,j) = \left\{\begin{array}{rl}
(1+\eps') \cdot \delta(i,j) & \text{if } i \neq \sigma(j, \loc) \\
\delta(i,j) & \text{otherwise}
\end{array}\right.
\]
Due to this \kmeans instance being $(1+\eps')$-perturbation stable, $\opt$ remains the unique optimum solution under these perturbed distances $\delta'$.
For any $\loc' \in \mathfrak{F}_k$, let $\cost'(\loc') = \sum_{j \in \cl} \min_{i \in \loc'} \delta'(i, j)^2$ be the cost of $\loc'$ under distances $\delta'$.
Partition the points in $\cl$ into the following groups:
\begin{itemize}
\item $\cl^1 = \{j \in \cl : \sigma(j, \loc) \in \loc-\opt \text{ and } \sigma(j, \opt) \in \loc \cap \opt\}$,
\item $\cl^2 = \{j \in \cl : \sigma(j, \loc) \in \loc\cap\opt \text{ and } \sigma(j, \opt) \in \opt-\loc\}$,
\item $\cl^3 = \{j \in \cl : \sigma(j, \loc), \sigma(j, \opt) \in \loc \cap \opt\}$, and
\item $\cl^4  = \{j \in \cl : \sigma(j, \loc) \in \loc-\opt \text{ and } \sigma(j, \opt) \in \opt - \loc\}$.
\end{itemize}
Observe $\cl^4 = \overline{\cl}_S$ (notation from the previous section) and that $\sigma(j, \opt) = \sigma(j, \loc)$ for $j \in \cl^3$.

As in the proof of Theorem \ref{thm:support}, let $c^*_j = \delta(j, \sigma(j, \opt))^2$ be the clustering cost incurred by point $j$ in the optimum solution and, analogously for $\loc$, $c_j = \delta(j, \sigma(j, \loc))^2$.
By considering the connection cost of each point on a case-by-case basis, we easily see
\begin{equation}\label{eqn:perturb}
\cost'(\opt) = \sum_{j \in \cl^1} (1+\eps')^2 \cdot c^*_j + \sum_{j \in \cl^2} \min\{(1+\eps')^2 \cdot c^*_j, c_j\} + \sum_{j \in \cl^3} c^*_j + \sum_{j \in \cl^4} (1+\eps')^2 \cdot c^*_j.
\end{equation}

Before putting all pieces together, we make one last observation. As $\loc$ is a nearly-good solution, $c^*_j \leq c_j$ for $j \in \cl^2$, and $c^*_j = c_j$ for $j \in \cl^3$, we have:
\begin{eqnarray*}
\sum_{j \in \cl^4} c_j & \leq & \sum_{j \in \cl^1} c_j + \sum_{j \in \cl^4}c_j = \cost(\loc) - \sum_{j \in \cl^2} c_j - \sum_{j \in \cl^3} c_j \\
& \leq & \cost(\loc) - \sum_{j \in \cl^2} c^*_j - \sum_{j \in \cl^3} c^*_j \leq \cost(\opt) + 2\eps \cdot \Psi(\loc)  - \sum_{j \in \cl^2} c^*_j - \sum_{j \in \cl^3} c^*_j \\
& = & \sum_{j \in \cl^1} c^*_j + \sum_{j \in \cl^4} c^*_j + 2\eps\left(\sum_{j \in \cl^4} c^*_j + c_j\right).
\end{eqnarray*}
\vspace{-1.5mm}
Rearranging,
\vspace{-1.5mm}
\begin{equation}\label{eqn:c4bound}
\sum_{j \in \cl^4} c_j \leq \frac{1}{1-2\eps}\left(\sum_{j \in \cl^1} c^*_j + (1+2\eps) \cdot \sum_{j \in \cl^4} c^*_j\right) \leq (1+6\eps) \cdot \left(\sum_{j \in \cl^1} c^*_j + \sum_{j \in \cl^4} c^*_j\right).
\end{equation}
\vspace{-3.5mm}
\begin{lemma}\label{lem:neargoodsol}
The nearly-good solution $\loc$ is the optimum solution.
\end{lemma}
\begin{proof}
Using \eqref{eqn:c4bound}, we bound $\cost(\loc)$ in the following way:
\begin{eqnarray}
\cost(\loc) & \leq & \cost(\opt) + 2\eps \cdot \Psi(\loc) = \cost(\opt) + 2\eps \sum_{j \in \cl^4} c^*_j + 2\eps \sum_{j \in \cl^4} c_j \nonumber \\
& \leq & \cost(\opt) + 2\eps \sum_{j \in \cl^4} c^*_j + 2\eps \cdot (1+6\eps) \cdot \left(\sum_{j \in \cl^1} c^*_j + \sum_{j \in \cl^4} c^*_j\right) \nonumber \\
& \leq & \sum_{j \in \cl^1} (1+6\eps) \cdot c^*_j + \sum_{j \in \cl^2} c^*_j + \sum_{j \in \cl^3} c^*_j + \sum_{j \in \cl^4} (1+6\eps) \cdot c^*_j  \nonumber\\
& \leq & \sum_{j \in \cl^1} (1+6\eps) \cdot c^*_j + \sum_{j \in \cl^2} \min\{(1+6\eps) \cdot c^*_j, c_j\} + \sum_{j \in \cl^3} c^*_j + \sum_{j \in \cl^4} (1+6\eps) \cdot c^*_j. \label{eqn:bound}
\end{eqnarray}
The last bound again uses $c^*_j \leq c_j$ for $j \in \cl^2$. Recall we chose $\eps$ so $(1+6\eps) = (1+\eps')^2$. Thus, combining \eqref{eqn:perturb}, \eqref{eqn:bound}
and the simple observation that $\cost'(\loc) = \cost(\loc)$ we see
\[ \cost'(\loc) = \cost(\loc) \leq \cost'(\opt). \]
Finally, because the instance is $(1+\eps')$-stable with $\opt$ being the unique optimum, it remains the unique optimum solution under the perturbed distances $\delta'$. This shows $\loc = \opt$. 
\end{proof}
We now conclude the proof of our main algorithmic result.
\begin{proof}[Proof of Theorem \ref{thm:alg}]
By Corollary \ref{cor:neargood}, within a polynomial number of iterations Algorithm \ref{alg:local} will have $\loc$ being a nearly-good solution. By Lemma \ref{lem:neargoodsol}, $\loc = \opt$.
Certainly Algorithm \ref{alg:local} will then terminate because there can be no improving swap for an optimal solution.
\end{proof}

We make the following interesting observation.
It states that the local-search algorithm provided earlier, when truncated to a polynomial number of steps,
provides a PTAS for arbitrary (not necessarily stable) instances of \kmeans and will fully solve stable instances 
whose stability constant is related to the $\eps$ in the PTAS approximation guarantee.

Consider fixed $0 < \eps \leq 1/6$ and fixed dimension $d$. Let $\rho := \rho(\eps, d)$ be as before (the constant in \cite{FRS16A,FRS16}).
\begin{corollary}
Consider an instance $\cl = (\cX, \fa, \delta)$ of \kmeans in a metric with doubling dimension $d$.
If we stop the loop of the $\rho$-swap local search heuristic in Algorithm \ref{alg:local} after $2k \cdot \ln(n \Delta)$ iterations,
then if $\cl$ is $\sqrt{(1+6\eps)}$-stable the algorithm will find the optimum solution and, otherwise,
the returned solution $\loc$ satisfies $\cost(\loc) \leq (1+6\eps) \cdot \cost(\opt)$.
\end{corollary}
\begin{proof}
We already argued it finds the optimum solution in $(1+\eps')$-stable instances where $\eps'$ satisfies $(1+\eps')^2 = 1+6\eps$. From Corollary \ref{cor:neargood}, which did not require the assumption that the instance
is stable, within $2k \cdot \ln(n \Delta)$ iterations some set $\loc$ considered in the algorithm satisfies $\cost(\loc) \leq \cost(\opt) + 2\eps \cdot \Psi(\loc)$.

Deriving the bound in \eqref{eqn:c4bound} also did not rely on stability, so
\[
\cost(\loc) \leq \cost(\opt) + 2\eps \cdot \Psi(\loc) \leq (1+2\eps) \cdot \cost(\opt) + 2\eps \cdot \cost(\loc).
\]
Rearranging, $\cost(\loc) \leq \frac{1+2\eps}{1-2\eps} \cdot \cost(\opt) \leq (1+6\eps) \cdot \cost(\opt).$
Thus, the final set returned by Algorithm \ref{alg:local} is at most this expensive: at most $(1+6\eps) \cdot \cost(\opt)$.
\end{proof}

\section{Roadmap of the Reduction for Theorem \ref{thm:hardness}}\label{sec:hardness0}
Our overall goal in the remaining sections is to prove Theorem \ref{thm:hardness}.
We remark that all of our reductions are deterministic reductions and run in polynomial time. The only randomization
in the reduction is in Theorem \ref{PCP-hypothesis} itself. Indeed, this seems essential given the current understanding of 
\uqsat as the only known hardness proofs are by randomized reductions.
Ultimately, by composing Theorem \ref{PCP-hypothesis} with our reductions, we obtain a randomized, polynomial time reduction from every language $L \in NP$
to \kmeans that has the following property. For every instance $I$ of $L$ we will have computed a value $c_I$ such that the resulting \kmeans instance $\cI_I$ has the following properties
depending on whether $I$ is a yes case or a no case.
\begin{itemize}
\vspace{-1.5mm}
\item[{\bf Yes case}:] With probability $\geq 1/{\rm poly}(|I|)$, $\cI_I$ is $s$-stable for some universal constant $s$ and the optimum solution to $\cI_I$ has cost $c_I$.
\vspace{-1.5mm}
\item[{\bf No case}:] Always, the optimum solution cost to $\cI_I$ is at least $\gamma \cdot c_I$ for some universal constant $\gamma>1$.
\end{itemize}
Given this, if there was an efficient $\gamma'$-approximation (for some $\gamma'<\gamma$) for $s$-stable instances of discrete \kmeans in $\mathbb R^d$ where $d$ is part of the input
then we could decide languages in $NP$ in the following way. By repeating the reduction polynomially many times and running the $\gamma'$-approximation
on each of the resulting \kmeans instances, with probability $\geq 1/2$  if $I$ was a yes instance then we would find some solution with cost $< \gamma \cdot c_I$
and, always, if $I$ was a no instance then every solution would have cost $\geq \gamma \cdot c_I$. That is, we would have decided $L$ with a randomized, polynomial
time algorithm with one-sided error (in the yes case) meaning $NP=RP$.

Starting with Theorem \ref{PCP-hypothesis}, we first reduce $Q$ parameter in \stqsatb to 3: that is we provide a stability-preserving reduction from \stqsatb to \stsatb\footnote{Paradise \cite{Orr} also presents a reduction from \stqsatb to the variant of \stsatb where every clause involves at most 3 variables. But we need
each clause to depend on exactly 3 distinct variables so we will present our approach here. We call it stable E3SAT-B}. We should point out that the letter \textsc{B} in \stsatb and \stqsatb (and other subsequent problems we consider) is to mean bounded occurrence but the exact bound may change as we reduce from one problem to another. We articulate these bounds in our reductions. For instance, in our reduction from an instance $\Phi$ of \stqsatb with bound $B$ we generate an instance $\Psi$ of \stsatb with bound $B'=\max\{7B,4BQ^2\}$ on the occurrence of each variable.

Our reduction from \stqsatb to \stsatb is a relatively simple reduction,
but it serves as a good introduction to the concept of preserving stability. Then we reduce \stsatb to stable instances of \textsc{3D-Matching} for some appropriate concept of stability
for this problem. Problem definitions, precise details of what we mean by stability for \textsc{3D-Matching}, and other finer-grained details we need to preserve will be discussed later. Finally, we reduce stable \textsc{3D-Matching}
instances to stable \kmeans instances to complete our proof.

\section{A Stability-Preserving Reduction From \stqsatb to \stsatb}\label{sec:hardness1}

Our first step is to show
hardness of \stsatb. There are standard reductions from \qsatb to \tsatb and,
if $Q$ is regarded as a constant, the most commonly-taught reduction is also an $L$-reduction. But more is needed to preserve stability, likely
the simple reduction the reader has in mind is not even parsimonious. While the reduction in this section is still quite simple, 
it serves as a warm-up to the concept of preserving stability in a reduction and it is a
necessary technical step toward our final goal. This is the first step of proof of hardness of stable instances of \kmeans.

Let $B, Q, s, \eps$ be constants from Theorem \ref{PCP-hypothesis}.
Let $\Phi$ be an instance of \textsc{QSAT-B}; we construct an 
instance $\Psi$ of \textsc{E3SAT-B} (3SAT-B with exactly 3 literals per clause) for some bound $B'$ on the number of occurrences of each variable which will depend only on $Q$ and $B$.
Properties of this reduction, including how it preserves stability, will be proven below.
Our reduction actually produces an instance of \stqsatb with {\em exactly} 3 literals per clause, hence it is actually
a reduction to stable instances of \textsc{E3SAT-B} and we require addressing clauses of size $< 3$ of the \stqsatb instance in the reduction.

Say $\Phi$ has variables $X$ and clauses $\mathcal C$ where each clause $C \in \mathcal C$ is viewed as a set of literals over distinct variables of $X$. We may depict a clause as, say, $x \vee \overline y \vee z$.
Before describing the reduction, consider the following gadget. For literals $\ell_1, \ell_2, \ell_3$ let $F(\ell_1, \ell_2, \ell_3)$ be the following collection of seven 3CNF clauses, applying the reduction $\overline{\overline x} \rightarrow x$ for any doubly-negated variable:
\begin{center}
$\overline \ell_1 \vee \ell_2 \vee \ell_3, ~~~~~ \overline \ell_1 \vee \ell_2 \vee \overline \ell_3, ~~~~~ \overline \ell_1 \vee \overline \ell_2 \vee \ell_3, ~~~~~ \overline \ell_1 \vee \overline \ell_2 \vee \overline \ell_3,$\\
$\ell_1 \vee \overline \ell_2 \vee \ell_3, ~~~~~ \ell_1 \vee \overline \ell_2 \vee \overline \ell_3, ~~~~~  \text{and } \ell_1 \vee \ell_2 \vee \overline \ell_3$.
\end{center}
One can easily check that the only way to satisfy all clauses in $F(\ell_1, \ell_2, \ell_3)$ is for all literals to be \textsc{False}.
So $F(\ell_1,\ell_2,\ell_3)$ enforces that all these three literals have to be false.

Our instance $\Psi$ of \textsc{3SAT-B} is constructed as follows. The variables of $\Psi$ consist of $X$ and a collection of new variables $Y$ we introduce below as we describe the clauses.
For each $C \in \mathcal C$, say $C = \{\ell^1_C, \ldots, \ell^{|C|}_{C}\}$.
We introduce some new variables $Y_C$ and clauses $\mathcal C'_C$ for $\Psi$.
\begin{itemize}
\item {\bf Case $|C| = 1$}. Let $Y_C = \{y_C, z_C\}$ and $\mathcal C'_C = F(\overline{\ell^1_C}, y_C, z_C)$.
\item {\bf Case $|C| = 2$}. Let $Y_C = \{w_C, y_C, z_C\}$ and $\mathcal C'_C = \{\ell^1_C \vee \ell^2_C \vee w_C\} \cup F(w_C, y_C, z_C)$.
\item {\bf Case $|C| = 3$}. Let $Y_C = \emptyset$ and $\mathcal C'_C = \{C\}$.
\item {\bf Case $|C| \geq 4$}. Let $Y_C = \{y^0_C, \ldots, y^{|C|}_C, z_C\}$ and $\mathcal C'_C$ be comprised of the following constraints:
\begin{enumerate}
\item $y^{i-1}_C \vee \ell^i_C \vee \overline{y^i_C}$ for each $1 \leq i \leq |C|$,
\item $\overline{y^{i-1}_C} \vee y^i_C \vee z_C$ for each $1 \leq i \leq |C|$,
\item $\overline{\ell^i_C} \vee y^{j-1}_C \vee z_C$ for each pair $1 \leq i < j \leq |C|$,
\item and $F(y^0_C, \overline{y^{|C|}_C}, z_C)$.
\end{enumerate}
\end{itemize}
Note the coarse upper bound of $|\mathcal C'_C| \leq 3Q^2$ and $|Y_C| \leq 2Q$ holds in each case.

Finally, the variables of $\Psi$ are $X' := X \cup Y$ where $Y=\bigcup_{C \in \mathcal C} Y_C$ 
and the constraints of $\Psi$ are $\mathcal C' := \cup_{C \in \mathcal C} \mathcal C'_C$. By construction, each
clause in $\Psi$ has exactly three literals over distinct variables. Also, each variable of $\Psi$ appears in at most $B' = \max\{7B, 4BQ^2\}$ clauses of $\Psi$.
The cases with $|C| = 1$ or 2 are simply padding gadgets to get exactly 3 literals per clause. For Case 4,
the constraints are to ensure that $z_C = \textsc{False}$, variables $y^0_C,\ldots,y^{|C|}_C$ are monotonic: set (i.e 
\textsc{False}, \textsc{False}, \textsc{False}, $\ldots$ , \textsc{True}, \textsc{True}) with
one switch from \textsc{False} to \textsc{True}, and that the switch 
from \textsc{False} to \textsc{True} appears at the first index $i$ where $\ell_i$
is \textsc{True}. This property is proven during the proof of Claim \ref{claim:pars}.

Let $n = |X|, m = |\mathcal C|, n' = |X'|$ and $m' = |\mathcal C'|$. Note,
\[ n' = n + \sum_{C \in \mathcal C} |Y_C| \leq n + (Q+1) \cdot m \leq n + (Q+1)B\cdot n \leq 2QB \cdot n 
\quad\text{and}\quad
 m' = \sum_{C \in \mathcal C} |\mathcal C'_C| \leq 3Q^2 \cdot m. \]

The following two claims are straightforward.
\begin{claim}\label{claim:pars}
Any satisfying assignment for $\Phi$ can be extended uniquely to a satisfying assignment for $\Psi$. 
Conversely, the restriction of any satisfying assignment for $\Psi$ to variables in $X$
is a satisfying assignment for $\Phi$.
\end{claim}
\begin{proof}
First, consider a satisfying assignment $\phi : X \rightarrow \{\textsc{True}, \textsc{False}\}$ for $\Phi$. For each clause $C \in \mathcal C$, we claim there is a unique way to assign values to $Y_C$ to satisfy all clauses in $\mathcal C'_C$.
This is simple to verify when $|C| \leq 3$, recalling the only way to satisfy all clauses in $F(u, v, w)$ 
is for the three literals to be false.

Suppose $|C| \geq 4$.
We know one of the literals $\ell^1_C, \ldots, \ell^{|C|}_C$ is true, let $i$ be the least-indexed literal that is true. Set $y^0_C, \ldots, y^{i-1}_C$ to \textsc{False}, $y^i_C, \ldots, y^{|C|}_C$
to \textsc{True}, and $z_C$ to \textsc{False}. This satisfies all clauses in $\mathcal C'_C$. We also claim this is the only way to satisfy all clauses.

To see this, note $F(y^0_C, \overline{y^{|C|}_C}, z_C)$ forces $y^0_C = z_C = \textsc{False}$ and $y^{|C|}_C = \textsc{True}$. Then the second set of constraints ensure the $y_C$-variables
are ``monotone'': that there exists some $1 \leq i' \leq |C|$ with $y^{j}_C = \textsc{False}$ for $j < i'$ and $y^j_C = \textsc{True}$ for $j \geq i'$.

It cannot be that $i' < i$, otherwise $y^{i'-1}_C \vee \ell^{i'}_C \vee \overline{y^{i'}_C}$ is \textsc{False} (recall $\ell^{i'}_C$ is \textsc{False} for $i' < i$).
It also cannot be that $i' > i$, otherwise $\overline{\ell^i_C} \vee y^{i'-1} \vee z_C$ is \textsc{False}. Thus, $i' = i$ meaning the assignment described above is the only way to extend the satisfying assignment for $\Phi$
to one that satisfies $\mathcal C'_C$.

The argument essentially reverses. Consider a satisfying assignment $\psi : X' \rightarrow \{\textsc{True}, \textsc{False}\}$ for
$\Psi$ and consider a clause $C \in \mathcal C$ of $\Phi$. If $|C| \leq 3$, it is easy to see the restriction of the assignment for $\Psi$ to $X$
satisfies $C$ itself. So suppose $|C| \geq 4$. As argued before, we must have both $y^0_C$ and $z_C$ being \textsc{False} and $y^{|C|}_C$ being \textsc{True}. Consider the least index $1 \leq i \leq |C|$ with $y^{i-1}_C$ being
\textsc{False} and $y^i_C$ being \textsc{True}. Then $\ell^i_C$ is \textsc{True}, so the original clause $C$ of $\Phi$ is satisfied.
\end{proof}

\begin{claim}\label{claim:lred}
Suppose at most $(1-\gamma) \cdot m$ clauses of $\Phi$ can be satisfied by any assignment for some $\gamma \geq 0$. Then at most $\left(1-\frac{\gamma}{3Q^2}\right)\cdot m'$ clauses of $\Psi$ can be satisfied
by any assignment.
\end{claim}
\begin{proof}
Consider a truth assignment $\psi$ for $\Psi$ and let $\phi$ be its restriction to $X$.
By Claim \ref{claim:pars}, any clause $C \in \mathcal C$ of $\Phi$ is satisfied by $\phi$ if all corresponding clauses in $\mathcal C'_C$ are satisfied by $\psi$.
This can happen for at most $(1-\gamma) \cdot m$ clauses $C \in \mathcal C$. So at least $\gamma \cdot m$ clauses $C \in \mathcal C$ have at least one corresponding clause in $\mathcal C'_C$ not being
satisfied.
Overall, the number of unsatisfied clauses in $\Psi$ is at least
$\gamma \cdot m \geq \frac{\gamma}{3Q^2}\cdot m'$.
\end{proof}

Finally, we show the reduction preserves stability. One might be tempted to think the analysis will be very similar to the {\bf no} case analysis. However, as stated in the introduction, we will later given an example showing
these concepts are fundamentally different by showing a well-known parsimonious $L$-reduction that does not preserve stability. This example is found in Appendix \ref{app:nostable}.

We have not attempted to optimize constants in our analysis below, but dependence on $B$ and $Q$ seems essential. The following shows that the stability drops only by at most a constant factor,
assuming $\Phi$ has bounded clause size and bounded occurrence for each variable.
\begin{theorem}
\label{thm:qsat_red}
Suppose $\Phi$ is $s$-stable. Then $\Psi$ is $s'$-stable where $s' = \frac{s}{8BB'Q^2}$.
\end{theorem}
\begin{proof}
Because $\Phi$ has a unique satisfying assignment $x^*$, then by Claim \ref{claim:pars}
there is some $y^*$ (assignment of values to variables in $Y$)
such that $(x^*, y^*)$ is the unique satisfying assignment for $\Psi$. Consider any truth assignment
$(x', y')$ for $\Psi$, we show the fraction of unsatisfied clauses in $\Psi$ is at least $s' \cdot \HW{(x^*, y^*)}{(x', y')}$.

Let $h = \HW{(x^*, y^*)}{(x', y')}$, $h_x = \HW{x^*}{x'}$, and $h_y = \HW{y^*}{y'}$.
Observe $|X| \cdot h_x + |Y| \cdot h_y = |X'| \cdot h$. Consider the following two cases.

~

\noindent
{\bf Case 1}: $|X| \cdot h_x \geq |X'| \cdot h/2$.\\
Because $\Phi$ is $s$-stable, then $x'$ leaves at least $m \cdot s \cdot h_x$ clauses of $\Phi$ unsatisfied. For each unsatisfied clause $C$ of $\Phi$, at least one clause of $\Psi$ in $\mathcal C'_C$
is not satisfied by $(x', y')$. So at least $m \cdot s \cdot h_x$ clauses of $\Psi$ are also not satisfied. Note the following:
\begin{itemize}
\item $|X| \leq Q \cdot m$, because each clause in $\Phi$ has at most $Q$ literals.
\item $B' \cdot |X'| \geq m'$, because each variable of $\Psi$ appears in at most $B'$ clauses.
\end{itemize}
The number of unsatisfied clauses of $\Psi$ is then at least
\[ m \cdot s \cdot h_x \geq \frac{s|X|}{Q} \cdot h_x \geq \frac{s|X'|}{2Q} \cdot h \geq \frac{s}{2B'Q}\cdot h \cdot m.'\]

~

\noindent
{\bf Case 2}: $|X| \cdot h_x < |X'| \cdot h/2$, equivalently $|Y| \cdot h_y > |X'| \cdot h/2$.\\
Let $\mathcal C_{bad} \subseteq \mathcal C$ be the clauses $C$ of $\Phi$ such that at least one variable in $Y_C$ is different between $y^*$ and $y'$. The number of $y$-variables that differ between $y^*$ and $y'$
is $|Y| \cdot h_y$ and $|Y_C| \leq 2Q$ for each $C \in \mathcal C$, so 

\begin{equation}\label{c-bad}
|\mathcal C_{bad}| \geq \frac{|Y|}{2Q} \cdot h_y.
\end{equation}

Partition $\mathcal C_{bad}$ into two groups:
\begin{itemize}
\item 
$\mathcal C_{bad}^1$: Clauses $C_i \in \mathcal C_{bad}$ such that at least one clause of $\Psi$ in $\mathcal C'_i$ is not satisfied by $(x', y')$.
\item $\mathcal C_{bad}^2$: Clauses $C_i \in \mathcal C_{bad}$ such that all clauses in $\mathcal C'_i$ are satisfied by $(x', y')$.
\end{itemize}

~

\noindent
{\bf Subcase 2.1}: $|\mathcal C_{bad}^1| \geq |\mathcal C_{bad}|/2$.\\
Then the number of clauses not satisfied by $(x', y')$ can be bounded from below as follows:
\[ |\mathcal C_{bad}^1| \geq \frac{|\mathcal C_{bad}|}{2} \geq \frac{|Y|}{4Q} \cdot h_y \geq \frac{|X'|}{8Q} \cdot h \geq \frac{h}{8B'Q} \cdot m', \]
where the 2nd inequality follows from (\ref{c-bad}).

\noindent
{\bf Subcase 2.2}: $|\mathcal C_{bad}^2| > |\mathcal C_{bad}|/2$.\\
Consider some $C \in \mathcal C_{bad}^2$. All constraints in $\mathcal C'_C$ are satisfied yet $y^*$ disagrees with $y'$ on $Y_C$, so it must be that $|C| \geq 4$ because, by construction of $\mathcal C'_C$,
the only way to satisfy all clauses in $\mathcal C'_C$ for $|C| \leq 3$ has all variables in $Y_C$ being \textsc{False}.

By construction of $\mathcal C'_C$ in the case $|C| \geq 4$, the fact that all clauses are satisfied means there is a unique index $1 \leq i \leq |C|$ with ${y'}^i_C = \textsc{True}$ and ${y'}^{i-1}_C = \textsc{False}$.
There is also a unique index $j$ with ${y^*}^j_C = \textsc{True}$ and ${y^*}^{j-1}_C = \textsc{False}$. But because $y^*$ and $y'$ disagree on $Y_C$, it must be that $i \neq j$.
Observe, then that $\ell^{\min\{i, j\}}_C$ has different values under $x^*$ and $x'$.
That is, there is some variable $x_\ell$ appearing in $C$ where $x^*_\ell \neq x'_\ell$.

The fact that each $C \in \mathcal C_{bad}^2$ witnesses at least one such variable $x_\ell$ with $x^*_\ell \neq x'_\ell$ and the fact that each variable of $\Phi$ appears in at most $B$ clauses means
there are at least $|\mathcal C_{bad}^2|/B$ variables $x_\ell \in X$ with $x^*_\ell \neq x'_\ell$.
That is:
\[ |X| \cdot h_x \geq \frac{|\mathcal C_{bad}^2|}{B} \geq \frac{|\mathcal C_{bad}|}{2B} \geq \frac{|Y|}{4BQ} \cdot h_y \geq \frac{|X'|}{8BQ} \cdot h \geq \frac{h}{8BB'Q} \cdot m'. \]
The fact that $\Phi$ is $s$-stable means at least $s \cdot h_x \cdot m$ clauses are not satisfied by $x'$. As before, for each clause $C \in \mathcal C$ not satisfied by $x'$
there is at least one clause in the group of clauses $\mathcal C'_C$ of $\Psi$
that is not satisfied by $(x', y')$. Thus, the number of clauses of $\Psi$ that are not satisfied by $(x', y')$ is at least
\[ s \cdot h_x \cdot m \geq s \cdot \frac{|X|}{Q} \cdot h_x \geq \frac{s}{8BB'Q^2} \cdot h \cdot m', \]
as required.
\end{proof}

\section{Reduction From \stsatbfull to \stdmbfull}\label{sec:hardness2}

In this section we show how to reduce stable instances of bounded 3SAT (\stsatb) to stable instances of bounded occurrence
3D matching.
We formally define the problems we consider in our reduction below.

\begin{definition}[\utdmbfull or \utdmb]
An instance of \utdmb problem is given via a hypergraph $G = (V_1 \cup V_2 \cup V_3, \, \cT)$ where, for $i = 1, \, 2, \, 3$, $\abs{V_i} = n$ , and each \emph{triple} $t \in \cT$ is of the form $(v_1, \, v_2, \, v_3)$ with $v_i \in V_i$. In the decision version of the problem, the task is to decide whether a \emph{perfect matching}, i.e., a subset $T^* \subseteq \cT$ of $n$ disjoint triples, exists given the guarantees that:
\begin{enumerate}
    \item Each vertex $v$ appears in at most $B$ triples, where $B$ is a constant.
    \item There is at most one perfect matching.
\end{enumerate}
\end{definition}

Throughout this section, for sets of triples $T^*$ and $T$, the Hamming weight function $\HW{T^*}{T}$ is defined as $\dfrac{\abs{T^* \Delta T}}{2 \abs{T^*}}$ for a nonempty set $T^*$.

\begin{definition}[\stdmbfull or \stdmb]
An instance of the \stdmb problem is an instance of the \utdmb problem that is $(s, \gamma)$-stable for $0 < s, \gamma < 1$, in the sense that it has the following guarantees.
\begin{enumerate}
    \item If a perfect matching $T^*$ exists, any subset of disjoint triples $T$ has a size at most $(1 - s \cdot \HW{T^*}{T}) n$, where $\HW{T^*}{T}$ is the fraction of the triples on which $T^*$ and $T$ disagree.
    \item If no perfect matchings exist, then any subset $T \subseteq \cT$ of disjoint triples has $\abs{T} \leq (1 - \gamma)\cdot n$.
\end{enumerate}
\end{definition}

\begin{definition}[\ucbtbfull]
An instance of the \ucbtb problem is given via the same hypergraph as in the \utdmb problem. In the decision version of the problem, the task is to decide whether a subset $T^* \subseteq \cT$ of size $n$
that covers all the $3n$ vertices exists given the guarantees that:
\begin{enumerate}
    \item Each vertex $v$ appears in at most $B$ triples, where $B$ is a constant.
    \item There is at most one set of triples solution covering all the vertices.
\end{enumerate}
Furthermore, if an instance $\cI$ of \ucbtb has a solution that covers all the nodes, we call it a \emph{covering} instance.
\end{definition}

\begin{definition}[\scbtbfull]
An instance of the \scbtb problem is an instance of the \ucbtb problem that is $(s, \gamma)$-stable for $0 < s, \gamma < 1$, in the sense that it has the following guarantees.
\begin{enumerate}
    \item If a unique covering solution $T^*$ exists, any subset of $n$ triples $T$ fails to cover at least an $s \cdot \HW{T^*}{T}$ fraction of the $3n$ vertices.
    \item If no covering solutions exist, then any subset $T \subseteq \cT$ of triples, $\abs{T} = n$, covers at most $(1 - \gamma)$ fraction of the points.
\end{enumerate}
\end{definition}

Here we show that the hardness of the \stsatb problem implies hardness for \stdmb. 

\begin{theorem}
There exists a polynomial reduction transforming any instance $\Psi$ of \stsatb with $n$ variables and $m$ clauses where each variable appears in at most $B$ clauses to an instance $\cI$ of \textsc{S-3DM-B} with $36 \cdot m$ points and $34 \cdot m$ triples where each point appears in at most 7 triples with the following properties
\begin{enumerate}
    \item \textbf{Yes Case:} if $\Psi$ is a satisfiable instance of \stsatb, then there exits a unique set of $12 \cdot m$ disjoint triples $T^*$  called the \emph{perfect matching}. Furthermore, for any disjoint
    set of triples $T$, it is true that $\abs{T} \leq (1 - s_1 \cdot \HW{T^*}{T}) \cdot \abs{T^*}$.
    \item \textbf{No Case:} if $\Psi$ is not satisfiable, then every set of disjoint triples has a size at most $(1 - \gamma_1)\cdot 12 \cdot m$,
\end{enumerate}
where $s_1$ and $\gamma_1$ are universal constants satisfying $0 < s_1, \gamma_1 < 1$ depending only on $B$ and the stability constant $s'$ from Theorem \ref{thm:qsat_red}.
\end{theorem}
We note that $\frac{n}{3} \leq m \leq \frac{B}{3} \cdot n$ holds for any \textsc{3SAT} instance where each clause contains exactly three different variables, this will be used frequently in the analysis of our reduction.


\subsection{The Reduction}
The reduction is inspired by known reductions (eg. \cite{gareyjohson}), but we change the clause gadget to make the reduction parsimonious. Throughout, we call a subset of triples $T$
a {\em packing} if the triples in $T$ are pairwise-disjoint.
The reduction uses two types of gadgets that are described below.

~

\noindent
{\bf Variable Gadgets}\\
This is a standard construction. For each variable $x_i$ of $\Psi$, we create a ``gear'', depicted in Figure \ref{fig:gear}. That is, we introduce {\em inner points} $u_i[b], w_i[b]$ and {\em tip points} $v_i[b], \overline{v_i}[b]$
for each $1 \leq b \leq \beta_i$ where $1 \leq \beta_i \leq B$ is the number of clauses of $\Psi$ that include variable $x_i$. The triples in this gadget are:
\[ G_i := \{\{v_i[b], u_i[b], w_i[b]\} : 1 \leq b \leq \beta_i\} \cup \{\{\overline{v_i}[b], w_i[b], u_i[b+1]\} : 1 \leq b \leq \beta_i\} \]
where we have used wraparound indexing (i.e. $u_i[\beta_i+1]$ is $u_i[1]$). Note $|G_i| = 2 \cdot \beta_j$.

\begin{figure}[H]
    \begin{center}
        \includegraphics[width = 8cm]{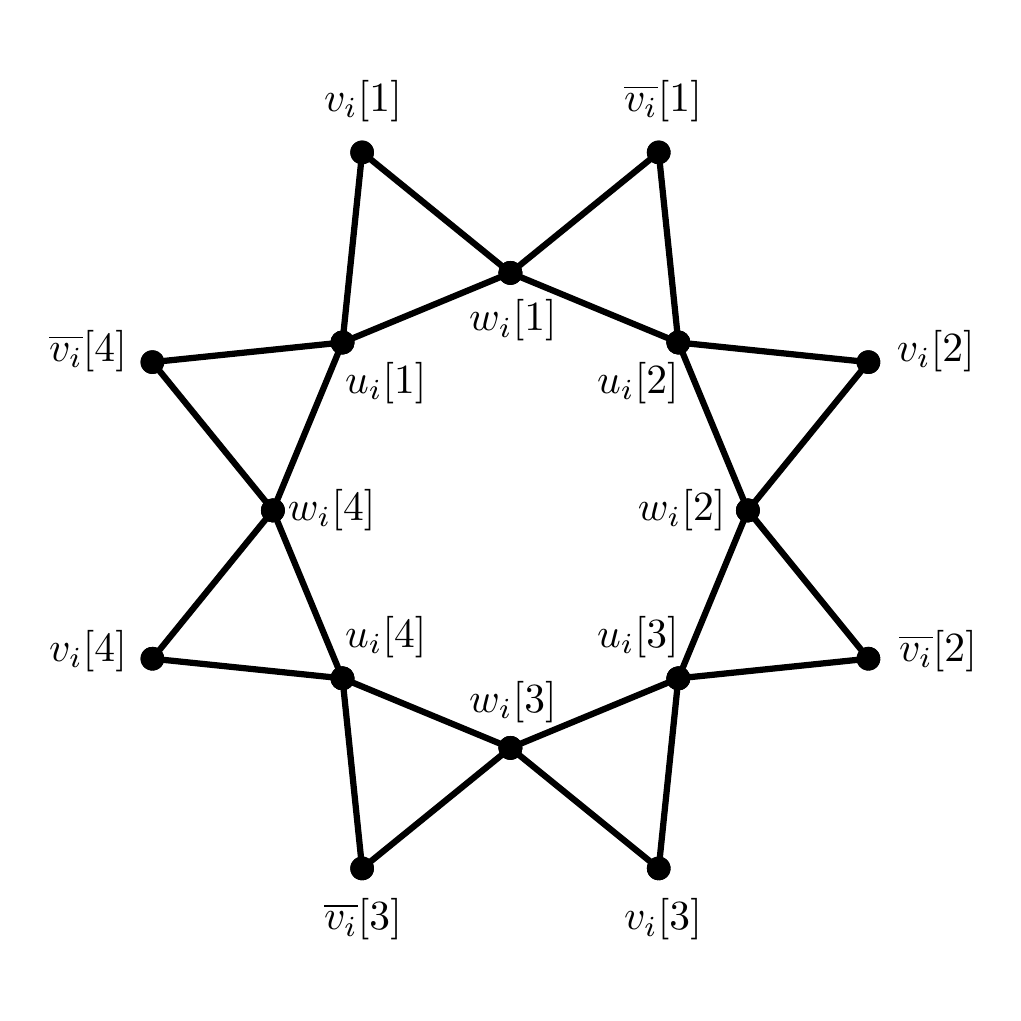}
    \end{center}
    \caption{The variable gadget for a variable $x_i$ with $\beta_i = 4$.}
    \label{fig:gear}
\end{figure}

Intuitively, setting $x_i$ to \texttt{True} corresponds to selecting all triples that include the points $\overline{v_i}[b], 1 \leq b \leq \beta_i$ and setting $x_i$ to \texttt{False} corresponds to selecting all triples that include the points
$v_i[b], 1 \leq b \leq \beta_i$. That is, the points that are not covered by such a packing of triples corresponds to the truth assignment being modelled.


\begin{claim} \label{claim:variable}
Let $T \subseteq G_i$ be any packing of triples. Then $|T| \leq \beta_i$ and $|T| = \beta_i$ if and only if $T = \{\{v_i[b], u_i[b], w_i[b]\} : 1 \leq b \leq \beta_i\}$ or $T = \{\{\overline{v_i}[b], w_i[b], u_i[b+1]\} : 1 \leq b \leq \beta_i\}$.
\end{claim}
\begin{proof}
First, note that since each $t \in G_i$ contains two inner points and there are $2 \cdot \beta_i$ inner points, then no packing $T \subset G_i$ has size more than $\beta_j$.
It is also clear that any packing of size exactly $\beta_i$ consists of alternating triples around the gear, i.e. it is one of the two packings from the statement of the claim.
\end{proof}

~

\noindent
{\bf Clause Gadgets}\\
Consider each clause $C = \ell_i \vee \ell_j \vee \ell_k$ involving variables $x_i, x_j, x_k$. Also let $b_i$ be such that $C$ is the $b_i$'th clause containing $x_i$ according to an arbitrary but fixed ordering of the occurrences of each variable.
Similarly define $b_j$ and $b_k$ for $x_j$ and $x_k$, respectively. Let $S_C$ be the set of ways to assign values to the variables $x_i, x_j, x_k$ to satisfy $C$, so $|S_C| = 7$.

We create 24 new points for the clause gadget. Three of them are {\em literal points} we call $y_i[b_i], y_j[b_j], y_k[b_k]$. Then, for $\alpha \in S_C$ we create three new vertices that we call {\em control points}: call these $z_i[\alpha], z_j[\alpha]$ and $z_k[\alpha]$.

There are 28 triples in this clause gadget that involve the literal points and control points for this gadget as well as the 6 tip points $v_i[b_i], \overline{v_i}[b_i], v_j[b_j], \overline{v_j}[b_j], v_k[b_k], \overline{v_k}[b_k]$ from the variable gadgets
corresponding to this particular occurrence of each variable in $C$.

For each $\alpha \in S_C$ let $v(\alpha, i)$ be the tip point $v_i[b_i]$ or $\overline{v_i}[b_i]$ corresponding to the truth assignment of $x_i$ under $\alpha$. That is, if $\alpha$ assigns \texttt{True} to $x_i$ then let $v(\alpha, i) = v_i[b_i]$ otherwise let $v(\alpha, i) = \overline{v_i}[b_i]$.

\[ G_C := \{\{z_i[\alpha], z_j[\alpha], z_k[\alpha]\} : \alpha \in S_C\} \cup \{\{v(\alpha, \ell), z_\ell[\alpha], y_\ell[b_\ell]\} : \alpha \in S_C,  \ell \in \{i,j,k\}\}. \]

A portion of this construction is depicted in Figure \ref{fig:clause}. The top layer of points are the control points for this particular $\alpha$ and do not appear in any other triples in the clause gadget (i.e. triples for different $\alpha' \neq \alpha$). The middle layer of points are the literal points: these appear in the other vertically-drawn triples for other portions of the clause gadget (i.e. for other $\alpha' \in S_C$). The bottom points are tip points from variable gadgets that correspond to this satisfying assignment: by leaving all three tips uncovered when we choose the triples from the corresponding variable gadgets, we are indicating to this clause that the variables have truth assignment $\alpha$ so we may pick the three vertically-drawn triples and cover all points in the figure.

\begin{figure}[H]
    \begin{center}
        \includegraphics[width = 6cm]{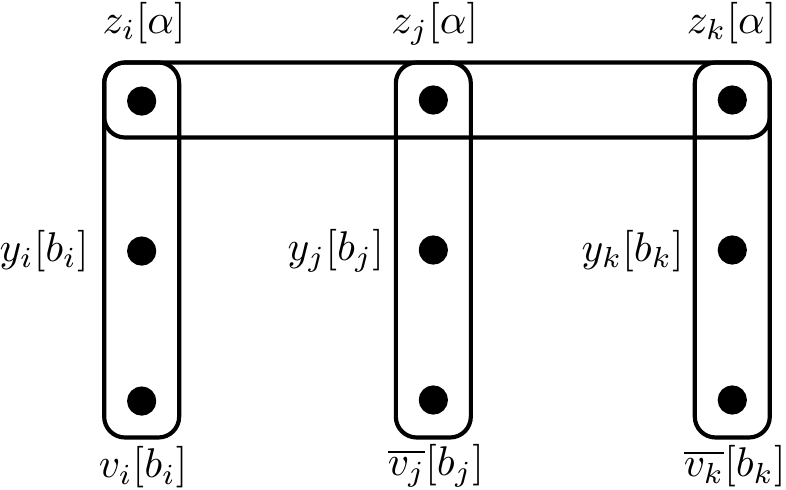}
    \end{center}
    \caption{Example of part of a clause gadget for a clause $C = \overline{x_i} \vee x_j \vee \overline{x_k}$ where $b_i, b_j$ and $b_k$ denote which occurrence of the corresponding variable appears in $C$. Here, $\alpha$ is the assignment $x_i = \texttt{True}, x_j = \texttt{False}, x_k = \texttt{False}$ which satisfies $C$. In particular, the tips of the gears that appear in this figure correspond to $\alpha$, not necessarily to the sign of the original variable in $C$.
    The rounded rectangles indicate the four triples in this gadget associated with $\alpha$. }
        \label{fig:clause}
\end{figure}

Important properties of this gadget are summarized in Claim \ref{clm:clause}.

\begin{claim}\label{clm:clause}
Let $T \subseteq G_C$ be a packing. Then $|T| \leq 9$. Furthermore, if $|T| = 9$ then for some $\alpha^* \in S_C$ we have
\[ T = \{\{z_i[\alpha], z_j[\alpha], z_k[\alpha]\} : \alpha \in S_C - \{\alpha^*\}\} \cup \{\{v(\alpha^*, \ell), z_\ell[\alpha^*], y_\ell[b_\ell]\} : \ell \in \{i,j,k\}\}. \]
\end{claim}
Intuitively, if $|T| = 9$ then the triples that cover the literal points all correspond to the same satisfying assignment $\alpha^*$.
\begin{proof}
If $T$ contains all seven triples that cover only control points (i.e. all triples in $\{\{z_i[\alpha], z_j[\alpha], z_k[\alpha]\} : \alpha \in S_C\}$) then it cannot contain any other triple in $G_C$ since they all include a control point.
If not, notice still that $T$ can only contain at most three triples that cover a literal point since there are only three literal points. In either case, $|T| \leq 9$.

Now suppose $|T| = 9$. Then $T$ contains exactly three triples that cover a literal point and exactly six triples that only cover control points. Let $\alpha^* \in S_C$ be the satisfying assignment such that $\{z_i[\alpha^*], z_j[\alpha^*], z_k[\alpha^*]\} \notin T$. Then all triples covering a literal point must be of the form $\{v(\alpha^*, \ell), z_\ell[\alpha^*], y_\ell[b_\ell]\}$ for the three choices $\ell \in \{i,j,k\}$, as required, since all other triples covering a literal point share a control point in common with the triples covering only control points.
\end{proof}

This completes the construction. Notice the number of points in the resulting instance is:
\[ \sum_{i=1}^n 4 \cdot \beta_i + \sum_{C} 24 = 12 \cdot m + 24 \cdot m = 36 \cdot m \]
and the number of triples is:
\[ \sum_{i=1}^n 2 \cdot \beta_i + \sum_{C} 28 = 6 \cdot m + 28 \cdot m = 34 \cdot m \]
where we have used the fact that $\sum_{i=1}^n \beta_i = 3 \cdot m$ as each clause involves precisely three distinct variables. Finally, it is easy to tell by inspection that every point lies in at most 7 triples: each inner point is in exactly 2 triples, tip points $v_i[b]$ and $\overline{v_i}[b]$ each lie in 1 triple from the variable gadget and at most 4 triples from the clause gadget corresponding to $b$'th occurrence of $x_i$, each control point is in exactly 2 triples, and each literal point is in exactly 7 triples. So this is an instance of \stdmb where each point appears in at most $7$ triples.


\subsection{Completeness Analysis}

{\bf The Canonical Packing $T^*$}\\
Here we suppose $\Psi$ has a unique satisfying assignment $x^*$ and every truth assignment $x'$ does not satisfy at least $s' \cdot HW(x^*, x') \cdot m$ clauses of $\Psi$ where $0 < s' < 1$ is an absolute constant.
We will show there is a single packing $T^*$ of size $12 \cdot m$ (i.e. it covers all points) and any other packing $T$ has $|T| \leq (1 - s_1 \cdot HW(T^*, T)) \cdot |T^*|$ for some absolute constant $0 < s_1 < 1$.

The construction of $T^*$ is straightforward. For each $i$, if $x^*_i = \texttt{True}$ then include all triples $\{\overline{v_i}[b], w_i[b], u_i[b+1]\}, 1 \leq b \leq \beta_i$
in $T^*$ and if $x^*_i = \texttt{False}$ then include all triples $\{v_i[b], u_i[b], w_i[b], 1 \leq b \leq \beta_i\}$ in $T^*$. The only points in the variable gadgets that are not yet covered
are those that correspond to the truth values of their corresponding variables.

Next consider a clause $C$. As in the construction of the clause gadget, say $x_i, x_j, x_k$ are the variables involved in $C$ and that, respectively, $b_i, b_j, b_k$ denotes which occurrence of the variable lies in $C$.
Let $\alpha^*$ be the truth assignment to $x_i, x_j, x_j$ given by $x^*$. Add the triples covering control points for $\alpha \neq \alpha^*$ to $T$, i.e. each of $\{z_i[\alpha], z_j[\alpha], z_k[\alpha]\}$ for $\alpha \in S_C - \{\alpha^*\}\}$.
Finally, add the three triples covering literal points corresponding to $\alpha^*$, namely $\{v(\alpha^*, \ell), z_\ell[\alpha^*], y_\ell[b_\ell]\}$.

By construction, every point in this \stdmb instance is covered by exactly one triple in $T^*$.

~

\noindent
{\bf Stability Analysis}\\
Next, consider some packing $T$ of triples. We show $|T| \leq (1 - s_1 \cdot HW(T^*, T)) \cdot |T^*|$ for some absolute constant $s_1$ to complete the stability analysis. We classify variables and clauses as follows.

Consider a variable $x_i$. 
\begin{itemize}
\item Call $x_i$ {\bf good} if $|T \cap G_i| = \beta_i$. From Claim \ref{claim:variable}, this means $T \cap G_i$ naturally corresponds to a truth assignment for $x_i$ (i.e. the uncovered tips correspond to the value of $x_i$).
\item Call $x_i$ {\bf bad} otherwise, so $|T \cap G_i| < \beta_i$.
\end{itemize}
Let $x'$ be the following truth assignment. Let $x'_i$ correspond to the truth assignment given by $T$ (i.e. corresponding to the uncovered tips) if $x_i$ is a good variable, otherwise let $x'_i = x^*_i$. Let $\zeta_B$ the fraction of bad variables, i.e. the number of bad variables is $\zeta_B \cdot n$.

Now consider a clause $C$.
\begin{itemize}
\item Call $C$ {\bf bad} if at least one of the three variables involved in $C$ is bad.
\item Call $C$ {\bf good-unsatisfied} if all three variables of $C$ are good, yet $x'$ does not satisfy $C$.
\item Call $C$ {\bf good-satisfied} if all three variables of $C$ are good and $x'$ satisfies $C$.
\end{itemize}

Next, we claim it suffices to assume that if $C$ is good-satisfied then we can assume that $|T \cap G_C| = 9$, i.e. the packing corresponds to a particular truth assignment in the canonical way described in Claim \ref{clm:clause}.

\begin{lemma}\label{lem:fixit}
Suppose for some constant $0 < s_1 \leq \frac{2}{17}$ we have the following: For every packing $T'$ with $|T' \cap G_C| = 9$ for each good-satisfied clause $C$ (here good-satisfied is with respect to $T'$)
we have $|T'| \leq (1-s_1 \cdot HW(T^*, T')) \cdot |T^*|$.
Then $|T| \leq (1 - s_1 \cdot HW(T^*, T)) \cdot |T^*|$ for every packing $T$.
\end{lemma}
\begin{proof}
By induction on the number of clauses $C$ such that $|T' \cap G_C| < 9$ for each good-satisfied clauses. The base case where there is no such clause $C$ is in fact our assumption.

Inductively, suppose $C$ is a good-satisfied clause yet $|T' \cap G_C| < 9$. Since all variables in $C$ are good and their corresponding values satisfy $C$, we can replace the triples in $T' \cap G_C$ with the 9 triples corresponding to this satisfying truth assignment $\alpha$
for $C$. Call this new packing $T''$. Notice $|T^* \Delta T'| \leq |T^* \Delta T''| + 17$ since we removed at most 8 triples and then added exactly 9 triples to form $T''$. By induction and recalling $HW(T^*, T) = |T^* \Delta T| / (2|T^*|)$, we see
\begin{eqnarray*}
|T'| & \leq & |T''|-1 \\
& \leq & \left(1 - s_1 \cdot \frac{|T^* \Delta T''|}{2 \cdot |T^*|} \right) \cdot |T^*| - 1 \\
& \leq & \left(1 - s_1 \cdot \frac{|T^* \Delta T'| - 17}{2 \cdot |T^*|} \right) \cdot |T^*| - 1 \\
& = & (1 - s_1 \cdot HW(T^*, T')) \cdot |T^*| + s_1 \cdot \frac{17}{2} - 1 \\
& \leq & (1 - s_1 \cdot HW(T^*, T')) \cdot |T^*|
\end{eqnarray*}
where the last bound follows by our assumption that $s_1 \leq 2/17$.
\end{proof}

Our goal is to prove the ``base case'' of Lemma \ref{lem:fixit} and articulate the constant $s_1$. We will now assume $|T \cap G_C| = 9$ for every good-satisfied clause $C$.
Continuing, we consider two cases. We have not attempted to ``balance'' the analysis between the two cases to optimize the resulting stability constant $s_1$.

In either case we use the observation that:
\[ |T| = \sum_{i=1}^n |T \cap G_i| + \sum_C |T \cap G_C| \leq \sum_{i=1}^n \beta_i + \sum_C 9 = 3m + 9m = 12m = |T^*|. \]
We will observe that throughout these cases, either $T$ is deficient in enough variable gadgets (i.e. $|T \cap G_i| \leq \beta_i - 1$ for enough $1 \leq i \leq n$) or deficient enough in clause gadgets (i.e. $|T \cap G_C| \leq 8$ for enough clauses $C$), where ``enough'' means the total deficiency is proportional $HW(T^*, T) \cdot |T^*|$.

\begin{itemize}
\item {\bf Case}: $\zeta_B > \frac{s'}{30 \cdot B^2} \cdot HW(T^*, T)$ where, recall, $0 < s' < 1$ is the stability constant for $\Psi$.\\

For each bad variable $x_i$, $|T \cap G_i| \leq \beta_i - 1$. So the number of variable gadgets for which $T$ has strictly fewer than the maximum possible number of triples
is at least
\[ \zeta_B \cdot n > \frac{s'}{30 \cdot B^2} \cdot HW(T^*, T) \cdot n \geq \frac{s'}{10 \cdot B^3} \cdot HW(T^*, T) \cdot m = \frac{s'}{120 \cdot B^3} \cdot HW(T^*, T) \cdot |T^*|. \]
where we have used $B \cdot n \geq 3 \cdot m$ and $|T^*| = 12 \cdot m$.

That is,
\[ |T| \leq \left(1-\frac{s'}{120 \cdot B^3} \cdot HW(T^*, T)\right) \cdot |T^*|. \]

\item {\bf Case}: $\zeta_B \leq \frac{s'}{30 \cdot B^2}  \cdot HW(T^*, T)$. \\

We begin by showing $HW(x^*, x')$ is at least proportional to $HW(T^*, T)$ in this case.
\begin{claim}
If $\zeta_B \leq \frac{s'}{30 \cdot B^2} \cdot HW(T^*, T)$, then $HW(x^*, x') \geq \frac{1}{5 \cdot B} \cdot HW(T^*, T)$.
\end{claim}
\begin{proof}
Suppose otherwise. We count $|T^* \Delta T|$ by considering different gadgets.
\begin{itemize}
\item Triples from bad variable gadgets. Each of $T$ and $T^*$ can have at most $B$ triples from any variable gadget, so at most $2B \cdot \zeta_B \cdot n$ triples come from bad variable gadgets.
\item Triples from bad clause gadgets. Each of $T$ and $T^*$ can have at most $9$ triples from any clause gadget and there are at most $B \cdot \zeta_B \cdot n$ bad clauses (as each variable participates in at most $B$ clauses), so at most $18B \cdot \zeta_B \cdot n$ triples from from bad clause gadgets.
\item Triples from good variable gadgets for a variable $x_i$ with $x'_i \neq x^*_i$. The number of such variables is exactly $HW(x^*, x') \cdot n$ by how we constructed $x'$ and each of $T$ and $T^*$ can have at most $B$ triples from variable gadgets, so at most $2B \cdot HW(x^*, x') \cdot n$ triples come from variable gadgets for good variables with $x'_i \neq x^*_i$.
\item Triples from clause gadgets that are not bad but contain a variable $x_i$ with $x'_i \neq x^*_i$. The number of such clauses is at most $B \cdot HW(x^*, x') \cdot n$ since there are $HW(x^*, x') \cdot n$ good variables that disagree on their assignment between $x^*$ and $x'$, and each variable appears in at most $B$ clauses. Each of $T$ and $T^*$ has at most 9 triples from any clause gadget, so at most $18 \cdot B \cdot HW(x^*, x) \cdot n$ triples come from  clause gadgets for clauses that are not bad but contain a variable $x_i$ with $x'_i \neq x^*_i$.
\end{itemize}
Recalling our assumption that $|T \cap T'_C| = 9$ for every good-satisfied clause $C$ and that $x^*$ satisfies every clause, we see that $T^*$ and $T$ agree on every other gadget. So we have accounted for all triples in $|T^* \Delta T|$.

From these four cases and using $n \leq 3 \cdot m, s' < 1$ and $|T^*| = 12 \cdot m$, we see
\begin{eqnarray*}
|T^* \Delta T| & \leq & 20B \cdot \zeta_B \cdot n + 20B \cdot HW(x^*, x') \cdot n \\
& \leq & 60B \cdot \zeta_B \cdot m + 60B \cdot HW(x^*, x') \cdot m \\
& = & 5B \cdot \zeta_B \cdot |T^*| + 5B \cdot HW(x^*, x') \cdot |T^*| \\
& \leq & \left(\frac{1}{6B} + 1\right) \cdot HW(T^*, T) \cdot |T^*|
\end{eqnarray*}
This contradicts $HW(T^*, T) = |T^* \Delta T|/(2|T^*|)$ which is true by definition.
\end{proof}

Thus, we have $HW(x^*, x') \geq \frac{1}{5 \cdot B} \cdot HW(T^*, T)$.
By stability of $\Psi$, the number of clauses of $\Psi$ not satisfied by $x'$ is at least $s' \cdot HW(x^*, x') \cdot m \geq \frac{s'}{5B} \cdot HW(T^*, T) \cdot m$.
The number of bad clauses is at most $B \cdot \zeta_B \cdot n$ since each bad variable appears in at most $B$ bad clauses. So the number of bad clauses is bounded by
\[ B \cdot \zeta_B \cdot n \leq 3B \cdot \zeta_B \cdot m \leq 3B \cdot \frac{s'}{30 \cdot B^2} \cdot HW(T^*, T) \cdot m = \frac{s'}{10 \cdot B} \cdot HW(T^*, T) \cdot m. \]

The number of good-unsatisfied clauses is at least the number of clauses not satisfied by $x'$ minus the number of bad clauses, thus is at least
\[  \frac{s'}{5 \cdot B} \cdot HW(T^*, T) \cdot m - \frac{s'}{10 \cdot B} \cdot HW(T^*, T) \cdot m = \frac{s'}{10 \cdot B} \cdot HW(T^*, T) \cdot m. \]

So in at least $\frac{s'}{10 \cdot B} \cdot HW(T^*, T) \cdot m = \frac{s'}{120 \cdot B} \cdot HW(T^*, T) \cdot |T^*|$ clause gadgets, $T$ has one fewer triple than $T^*$. That is,
\[ |T| \leq \left(1 - \frac{s'}{120 \cdot B} \cdot HW(T^*, T)\right) \cdot |T^*|. \]
\end{itemize}

In either case, we have $|T| \leq (1 - s_1 \cdot HW(T^*, T)) \cdot |T^*|$ where $s_1 = \frac{s'}{120 \cdot B^3}$. Clearly $s_1 \leq \frac{2}{17}$ also holds, which is required for Lemma \ref{lem:fixit}.
This completes the stability analysis.

\subsection{Soundness Analysis}
Suppose any truth assignment to $\Psi$ satisfies at most $(1-\rho) \cdot m$ clauses where $\rho > 0$ is an absolute constant. Our analysis breaks into two cases similar to the stability analysis, but it is much shorter.

Let $T$ be any packing of triples. We define good and bad variables as in the stability analysis and let $\zeta_B$ denote the fraction of bad variables. Let $x'$ be the following  truth assignment. We set $x'_i$ to be the truth assignment
corresponding to the packing $T$ for every good variable $x_i$. We set $x'_i$ arbitrarily for every bad variable. Also let $\Lambda = 12 \cdot m$, which is the number of triples required to cover every point exactly once (i.e. the optimal value in the completeness case).

\begin{itemize}
\item {\bf Case}: $\zeta_B > \frac{\rho}{6B}$.\\

As in the analogous case from the stability analysis and using $B\cdot n \geq 3m$.
\[ |T| \leq \Lambda - \zeta_B \cdot n \leq \Lambda - \frac{\rho}{6B} \cdot n \leq \Lambda - \frac{\rho}{2B^2} \cdot m = \left(1 - \frac{\rho}{24B^2}\right) \cdot \Lambda. \]
Again, this is because in the $\zeta_B \cdot n$ variable gadgets for bad variables $x_i$, we have $|T \cap G_i| \leq \beta_i-1$.

\item {\bf Case}: $\zeta_B \leq \frac{\rho}{6B}$.\\

The number of clauses that involve a bad variable is at most $B \cdot \zeta_B \cdot n$. Also, the number of clauses not satisfied by $x'$ is at least $\rho \cdot m$. Of these,
at least $\rho \cdot m - B \cdot \zeta_B \cdot n$ clauses include only good variables. Since
these clauses are not satisfied, by Claim 
\ref{clm:clause}, $T$ includes at most 8 triples from the corresponding clause gadget. That is,
using our bound on $\zeta_B$ and $n \leq 3m$ we have

\[ |T| \leq \Lambda - (\rho \cdot m - B \cdot \zeta_B \cdot n) \leq \Lambda - \frac{\rho}{2} \cdot m =  \left(1 - \frac{\rho}{24}\right) \cdot \Lambda. \]
\end{itemize}
Considering either case, we see $|T| \leq (1 - \rho/(24B^2)) \cdot \Lambda$, as required.

\section{From \stdmbfull to \scbtbfull}\label{sec:hardness3}

Here, we show a reduction from any instance of the \stdmb problem to an instance of the \scbtb problem. In this section, for two sets of triples $T$ and $T'$, we define the Hamming distance function the same as before, $\HW{T}{T'} = \dfrac{\abs{T \Delta T'}}{2n}$, which is the size of the symmetric distance of the two set of triples divided by two times the size of the maximum set of disjoint triples.

\begin{theorem}\label{theo:sctb}
There exists a polynomial reduction transforming any instance $\cI$ of \stdmb with $3n$ vertices and $m$ triples to an  instance $\cI'$ of \scbtb with the same number of vertices and triples, such that
\begin{enumerate}
    \item \textbf{Yes Case:} if $\cI$ admits a perfect matching, then there exists a unique set of $n$ triples $T^*$ in $\cI'$ that cover the entire set of nodes. Furthermore, for any set of triples $T$ of size $n$ (different
    from $T^*$), $T$ covers at most $(1 - s_2 \cdot \HW{T^*}{T}) \cdot 3n$ vertices.
    \item \textbf{No Case:} if $\cI$ does not admit a perfect matching, then every set of triples $T$ with $\abs{T} = n$, covers at most $(1 - \gamma_2)\cdot 3n$ vertices, 
\end{enumerate}
where $\gamma_2$ ($< \gamma_1)$ and $s_2$ are universal positive constants in the $(0, 1)$ interval.
\end{theorem}

\begin{proof}

The transformation function on the instances is an identity function, that is, we consider the same graph with the same set of triples. \\

\noindent
\textbf{Completeness:} Let $T^*$ be the perfect matching for instance $\cI$. Obviously, the same set $T^*$ covers the entire set of vertices in $\cI'$ as well. Any set of triples that covers all the vertices and has a size of $n$ must be disjoint (since there are $3n$ nodes to cover), hence a perfect matching. By the uniqueness of perfect matching in $\cI$, we conclude that the covering set of triples for $\cI'$ is also unique. Now, consider a set of triples $T$ with $\abs{T} = n$, and let $T'$ be a maximal subset of disjoint triples of $T$. Let $cov(T)$ denote the number of vertices covered by $T$. Note that any of the triple in $T - T'$ intersect with at least one other triple in $T'$ by maximality of $T'$, so they can cover at most $2$ extra vertices compared to the vertices already covered by $T'$.  Then
\begin{align}
cov(T)	& \leq 3\abs{T'} + 2\abs{T - T'} \nonumber \\
		& \leq 3(1 - s_1 \cdot \HW{T^*}{T'})n + 2\abs{T - T'} \nonumber \\
		& = 3n - \dfrac{3s_1 \cdot \abs{T^* \Delta T'}}{2n} n + 2\abs{T - T'} \nonumber \\
		& = 3n - \dfrac{3s_1}{2} \abs{T^* \Delta T'} + 2\abs{T - T'} \nonumber \\
		& \leq 3n - \dfrac{3s_1}{2} \abs{T^* \Delta T} + \dfrac{3s_1}{2} \abs{T - T'} + 2\abs{T - T'} \nonumber \\
		\label{eqn:covergae}
		& = 3n - \dfrac{3s_1}{2} \abs{T^* \Delta T} + \dfrac{3s_1 + 4}{2} \abs{T - T'}
\end{align}
where in the last inequality, we have used the fact that $\abs{T^* \Delta T'} \geq \abs{T^* \Delta T} - \abs{T - T'}$. Now, we consider two cases:\\

\noindent
\textbf{Case 1:} If $\dfrac{3s_1}{2} \abs{T^* \Delta T} - \dfrac{3s_1 + 4}{2} \abs{T - T'} > \dfrac{3s_1}{4} \abs{T^* \Delta T}$. In this case, using \eqref{eqn:covergae} we simply bound $cov(T)$ as
\begin{align*}
cov(T)	& \leq  3n - \dfrac{3s_1}{2} \abs{T^* \Delta T} + \dfrac{3s_1 + 4}{2} \abs{T - T'} \\
		& < 3n - \dfrac{3s_1}{4} \abs{T^* \Delta T} \\
		& = 3n - \dfrac{3s_1 \cdot n}{4n} \abs{T^* \Delta T} \\
		& = 3n \left( 1 - \dfrac{s_1 \cdot \HW{T^*}{T}}{2} \right).
\end{align*}

\noindent
\textbf{Case 2:} If $\dfrac{3s_1}{2} \abs{T^* \Delta T} - \dfrac{3s_1 + 4}{2} \abs{T - T'} \leq \dfrac{3s_1}{4} \abs{T^* \Delta T}$, with re-arranging the terms we get $\abs{T - T'} \geq \dfrac{3s_1}{6s_1 + 8} \abs{T^* \Delta T}$. 
Since every triple in $T - T'$ can cover at most $2$ extra nodes compared to $T'$, each triple in this set can be 
charged with one deficiency in the coverage of $T$ from a maximum of $3n$. Therefore, it must be the case that 
\begin{align*}
cov(T) 	& \leq 3n - \abs{T - T'} \\
		& \leq 3n - \dfrac{3s_1}{6s_1 + 8} \abs{T^* \Delta T} \\
		& = 3n - \dfrac{3s_1 \cdot n}{(3s_1 + 4) \cdot 2n} \abs{T^* \Delta T} \\
		& = 3n \left( 1 - \dfrac{s_1}{3s_1 + 4} \right) \HW{T^*}{T}.
\end{align*}
Choosing $s_2 = \frac{s_1}{3s_1 + 4}$ is sufficient in both cases to obtain the required stability condition.\\

\noindent
\textbf{Soundness:} Let $T$ be a set of $n$ triples in $\cI'$, and assume $T'$ is a maximal disjoint subset of $T$. From the No case of the \stdmb problem, we know that $\abs{T'} \leq (1 - \gamma_1)n$, thus $\abs{T - T'} \geq \gamma_1 \cdot n$. As argued before, the size of $T - T'$ is a lower bound on the number of nodes that are not covered. Therefore
\begin{align*}
cov(T)	& \leq 3n - \abs{T - T'} \\
		& \leq 3n - \gamma_1 \cdot n \\
		& = 3n \left(1 - \dfrac{\gamma_1}{3} \right). 
\end{align*}
Choosing $\gamma_2 = \dfrac{\gamma_1}{3}$ yields the required result.
\end{proof}

\section{Reduction From \scbtbfull to \textsc{Stable} \kmeans}\label{sec:hardness4}

We now show a reduction from any instance of the \scbtb to $(1+\eps_0)$-stable instances of the discrete \kmeans problem,
hence completing proof of Theorem \ref{thm:hardness}. 
In the following, by the Hamming weight function of two sets of centers $C$ and $C'$ we mean $\dfrac{\abs{C \Delta C'}}{2k}$. When $\abs{C} = \abs{C'}$, we have $\abs{C \Delta C'} = \abs{C - C'} + \abs{C' - C} = 2\abs{C - C'} = 2\abs{C' - C}$, so the Hamming weight function effectively becomes $\HW{C}{C'} = \dfrac{\abs{C - C'}}{k}$.

\begin{theorem}\label{thm:scbtb-kmeans}
There is a polynomial-time reduction transforming any instance $\cI$ of \scbtb to an instance $(\cX, \cC, \delta)$ of the (discrete) \kmeans problem in Euclidean space with the following properties.
\begin{enumerate}
    \item \textbf{Yes Case:} If $\cI$ is an instance from the {\bf Yes} case in Theorem \ref{theo:sctb}, then there is a unique optimum set of $k$ centres $C^* \subset \cC$ whose cost is $6k$ which remains the unique optimum solution after any $(1+\epsilon_\circ)$ perturbations.
    Furthermore, for any set of centres $C$ of size $k$, the \kmeans cost (without perturbations) of $C$ is at least $(6 + s_2 \cdot \HW{C^*}{C}) \cdot k$.
    \item \textbf{No Case:} If $\cI$ is an instance from the {\bf No} case in Theorem \ref{theo:sctb}, then every set of $k$ centres $C \subset \cC$, has a \kmeans cost of at least $(1+\gamma_2) \cdot 6k$.
\end{enumerate}
where $\epsilon_\circ > 0$ is some universal constant and $\gamma_2, s_2 > 0$ is the constant from Theorem \ref{theo:sctb}.
\end{theorem}
The last statement of the {\bf Yes} case is not required to conclude the proof Theorem \ref{thm:hardness}, but it may prove useful for future stability-preserving reductions that build off of our reduction.

\begin{proof}
For every point $v_j \in V$, $j = 1, \, 2, \, \ldots, \, 3n$, create a point of $\cX$, $x_j = \boldsymbol{e_j} \in \bR^{3n}$, where $\boldsymbol{e_j}$ has a $1$ in its $j^{th}$ coordinate, and zeros everywhere else. For every triple $t_\ell = (v_i, \, v_j, \, v_k) \in \cT$, $\ell = 1, \, 2, \, \ldots, \, m$, create a new point $y_\ell = \boldsymbol{e_i} + \boldsymbol{e_j} +\boldsymbol{e_k}$ in $\cC$. Finally, use $k = n$ as the number of centres for the resulting \kmeans instance.

We say that  $y_\ell$ \emph{covers} a point $x_i$ if $y_\ell$ has a 1 in its $i^{th}$ coordinate.
Note for two points $x_i, y_\ell$ we have $\delta(x_i, y_\ell)^2 = 2$ if $y_\ell$ covers $x_i$ and $\delta(x_i, y_\ell)^2 = 4$ otherwise.

~

\noindent
\textbf{Completeness:} We use $\epsilon_\circ := \min\{0.4, \frac{2s_2}{3s_2 + 3}\}$.

Assume that, in a given instance $\cI$ of the \scbtb, there exists a unique covering set of triples $T^*$, and all other subsets $T \subseteq \cT, |T| = n$ cover a fraction less than $1-s_2 \cdot \HW{T^*}{T}$ of the points where $s_2$ is the constant from Theorem \ref{theo:sctb}.
Let $C^* \subset \cC$ be the set of centres corresponding to the $n$ triples in $T^*$, i.e. $C^* = \{\boldsymbol{e_i} + \boldsymbol{e_j} +\boldsymbol{e_k} : (v_i, v_j, v_k) \in T^*\}$. Since $T^*$ covers all the points in instance $\cI$, each $x_j \in \cX$ is covered by some centre in $C^*$ so $\delta(x_j, C^*)^2 = 2$. 
Thus the cost of $C^*$ is $2 \cdot |\cX| = 6n$.

We next show the cost of any set of centres $C \subset \cC, |C| = k$ is at least $(6 + s_2 \cdot \HW{C^*}{C}) \cdot k$. To that end, let $T \subset \cT$ be the triples corresponding to centres in $C$. Observe $T$ leaves at least $s_2 \cdot \HW{T^*}{T} \cdot 3n$ points not covered in $\cI$. The corresponding points in $\cX$
are not covered by $C$ and $\delta(v_j, C)^2 = 4$ for any $v_j$ not covered by $C$. Thus, if we let $\eta_C$ denote the fraction of points in $\cX$ not covered by $C$ (noting $\eta_C \geq s_2 \cdot \HW{T^*}{T}$) we have:
\[ \sum_{j \in \cX} \delta(j, C)^2 = (1-\eta_C) \cdot 2 \cdot 3n + \eta_C \cdot 4 \cdot 3n = (1 + \eta_C) \cdot 6n \geq (1 + s_2 \cdot \HW{C^*}{C}) \cdot 6n, \]
as required (noting $\HW{C^*}{C} = \HW{T^*}{T}$).


Now let $\delta'()$ be any $(1+\epsilon_0)$-perturbation of $\delta$, so $\delta(j,i) \leq \delta'(j,i) \leq (1+\epsilon_0) \cdot \delta(j,i)$ for all $j \in \cX, i \in \cC$. We claim $C^*$ remains the unique optimum solution to this \kmeans instance. To that end, consider any $C \subseteq \cC, |C| = k$ with $C \neq C^*$.
Let $\sigma : \cX \rightarrow C$ map each point to its nearest centre according to distances $\delta'()$. Also let $\sigma^* : \cX \rightarrow C^*$ map each point to its nearest centre in $C^*$, again under distances $\delta'()$. Since each $x_j$ is covered by exactly one centre in $C^*$ and since
the perturbation $\epsilon_\circ$ is small (i.e. $2 \cdot (1+\epsilon_\circ)^2 < 4$), $\sigma^*(x_j)$ is the unique centre in $C^*$ covering $x_j$.

Note as $\epsilon_\circ$ is small (i.e. $2 \cdot (1+\epsilon_\circ)^2 < 4$ which holds because $\epsilon_\circ < 0.4$), if a point $x_j \in \cX$ is covered by some centre in $C$, then $\sigma$ will map $x_j$ to a point in $C$ that covers $x_j$.
We partition $\cX$ into three groups:
\begin{itemize}
\item $\mathbf{A}$: The points $x_j$ not covered by any centre in $C$. Note $C \neq C^*$ means $\mathbf{A} \neq \emptyset$.
\item $\mathbf{B}$: The points $x_j$ covered by a centre in $C$ and have $\sigma(x_j) \in C-C^*$.
\item $\mathbf{R}$: The points $x_j$ covered by a centre in $C$ and have $\sigma(x_j) \in C \cap C^*$.
\end{itemize}
For each $x_j \in \cX$, let $\epsilon_j$ be such that $\delta'(x_j, \sigma^*(x_j)) =  \delta(x_j, \sigma(x_j)) \cdot (1+\epsilon_j)$, so $0 \leq \epsilon_j \leq \epsilon_\circ$.
Note the following bounds on $\delta(x_j, \sigma(x_j))$ for each point $x_j \in \cX$:
\begin{itemize}
\item For $x_j \in \mathbf{A}$, $\delta'(x_j, \sigma(x_j))^2 \geq 4$ since $C$ does not cover $x_j$.
\item For $x_j \in \mathbf{B}$, $\delta'(x_j, \sigma(x_j))^2 \geq 2$.
\item For $x_j \in \mathbf{R}$, $\delta'(x_j, \sigma(x_j))^2 = 2 \cdot (1+\epsilon_j)^2$. This is because $\sigma(x_j) = \sigma^*(x_j)$ for such $j$ as there is a unique centre (namely $\sigma^*(j)$) in $C \cap C^*$ covering $x_j$.
\end{itemize}
Therefore:
\[ cost'(C) := \sum_{x_j \in \cX} \delta'(x_j, C)^2 \geq |\mathbf{A}| \cdot 4 + |\mathbf{B}| \cdot 2 + \sum_{x_j \in \mathbf R} 2 \cdot (1+\epsilon_j)^2. \]
On the other hand, using $(1+\epsilon_\circ)^2 < 1 + 3\epsilon_\circ$ (as $\epsilon_\circ < 1$)  we see:
\[ cost'(C^*) := \sum_{x_j \in \cX} \delta'(x_j, C^*)^2 < (|\mathbf{A}| + |\mathbf{B}|) \cdot 2 \cdot (1+3 \cdot \epsilon_\circ) + \sum_{x_j \in \mathbf R} 2 \cdot (1+\epsilon_j)^2. \]
The bound is strict since $\mathbf{A} \neq \emptyset$.
From what we showed earlier, $|\mathbf{A}| \geq 3n \cdot s_2 \cdot \HW{C}{C^*}$ since $\mathbf{A}$ are the points not covered by $C$. Also, $|\mathbf{B}| \leq 3 \cdot |C-C^*| = 3k \cdot \HW{C}{C^*}$.
We now see
\begin{eqnarray*}
 cost'(C) - cost'(C^*)
 & > & |\mathbf{A}| \cdot (2 - 3\cdot \epsilon_\circ) - |\mathbf{B}| \cdot 3 \cdot \epsilon_\circ \\
 & \geq & 3n \cdot s_2 \cdot \HW{C}{C^*} \cdot (2 - 3\cdot \epsilon_\circ) - 3n \cdot \HW{C}{C^*} \cdot \epsilon_\circ \\
 & = & 3n \cdot \HW{C}{C^*} \cdot (2s_2 - (3s_2 + 3) \cdot \epsilon_\circ) = 0.
\end{eqnarray*}
The equality with 0 is by our choice of $\epsilon_\circ$. Therefore, $cost'(C) > cost'(C^*)$ as required. Since this holds for any $C \neq C*$ with $|C| = n$, we see $C^*$ remains the unique optimum set of $n$ centres in the instance with perturbed distances $\delta'$.

~

\noindent
\textbf{Soundness:} Assume $\cI$ is a non-covering instance of \scbtb, so we have that any $T \subseteq \cT, |T| = n$ can cover at most $(1 - \gamma_2) \cdot 3n$ points. In the \kmeans instance, consider any $C \subseteq \cC, |C| = k$. Let $T$ be the triples corresponding to $T$. Since $T$ fails to cover $\eta'_T \geq \gamma_2 \cdot 3n$ points,
then the cost of $C$ is:
\[ 2 \cdot (3n-\eta'_T) + 4 \cdot \eta'_T  = 6n + 2\eta'_T \geq (1 + \gamma_2) \cdot 6n. \]
\end{proof}


\begin{proof}[Proof of Theorem \ref{thm:hardness}]
Now we combine the results of Sections \ref{sec:hardness1} to \ref{sec:hardness4}.
Assuming Theorem \ref{PCP-hypothesis}, it follows that \stqsatb is hard. The reduction in Section \ref{sec:hardness1}
implies the hardness of \stsatb. Then reduction from Section \ref{sec:hardness2} implies hardness of \stdmb. Reduction of
Section \ref{sec:hardness3} implies hardness of \scbtb. Finally, Theorem \ref{thm:scbtb-kmeans} implies hardness of
stable \kmeans.
\end{proof}

\section{Conclusion}
We showed stable instances of \kmed and \kmeans in metrics with constant doubling dimension, including constant-dimensional Euclidean metrics, can be solved in polynomial time by using a standard local search algorithm that always takes the best improvement. We also showed stable instances are hard to solve for some stability constant in arbitrary dimension Euclidean metrics. A natural problem is to find faster algorithms for solving stable \kmeans.
A related direction to consider is what notions of stability cause other {\bf PLS}-complete problems to become polynomial-time solvable.

We also used the concept of stability-preserving reductions to show hardness for stable \kmeans by leveraging a new PCP construction by Paradise \cite{Orr}. What other stable optimization problems can be proven hard with this approach?


\bibliographystyle{plain}
{
\bibliography{stable-clustering}
}


\appendix

\section{Alternative Local Search Convergence Analysis}\label{app:alt}

Consider some problem where $\mathfrak F$ is the set of feasible solutions and each $\loc \in \mathfrak F$ is endowed with an {\em integer} value $\cost(\loc)$ that can be evaluated in polynomial time.
The goal is to find some $\loc \in \mathfrak F$ with minimum $\cost(\loc)$. We describe a setting encountered in many approximation algorithms based on local search and provide alternative analysis on the convergence of the
local search heuristics that show the locality guarantee is obtained after a polynomial number of iterations. That is, we avoid the ``$\eps$'' that is typically lost in the guarantee from many local search algorithm
that only take noticeable improvements ({\em e.g.} only if the cost improves by a $(1-\eps/k)$-factor for some value $k$). When we say ``polynomial time'' in this context, we mean the running time
is polynomial in the input size of the underlying problem.

For each $\loc \in \mathfrak F$, let $\mathcal N(\loc) \subseteq \mathfrak F$
be a set of {\em neighbouring solutions} with the property $\loc \in \mathcal N(\loc)$. We assume $\mathcal N(\loc)$ can be enumerated in polynomial time (implying $|\mathcal N(\loc)|$ is polynomially-bounded).
Suppose we also know some $\Delta$ such that $\cost(\loc) \leq \Delta$ for each $\loc \in \mathfrak F$ where $\log \Delta$ is polynomially bounded in the input size.
If the reader wants to consider a specific setting, consider the \kmed problem where $\mathfrak F$ is all subsets of $k$ centres and $\mathcal N(\loc)$
is the set of feasible solutions $\loc'$ with $|\loc - \loc'|=1$ (i.e. the single-swap setting), and $\Delta = |\cl| \cdot \max_{i,j} \delta(i,j)$.

Next, suppose we have the following ``locality'' analysis: for each $\loc \in \mathfrak F$ there is a set of neighbouring solutions $\mathfrak G(\loc) \subseteq \mathcal N(\loc)$ (sometimes called {\em test swaps}) where
\begin{equation}\label{eqn:step}
\sum_{\loc' \in \mathfrak G(\loc)} \left(\cost(\loc') - \cost(\loc)\right) \leq \alpha \cdot \cost(\opt) - \beta \cdot \cost(\loc).
\end{equation}
for some fixed rational values $\alpha \geq 1$ and $0 < \beta \leq 1$ that are both integer multiples of $M$ (i.e. the least-common multiple of the denominators of $\alpha$ and $\beta$).
Finally, suppose we have a bound $|\mathfrak G(\loc)| \leq \kappa$ for each $\loc \in \mathfrak F$ on the number of ``test swaps'' in the above bound. We know $\kappa$ is polynomially-bounded because
$|\mathcal N(\loc)|$ is polynomially-bounded, but perhaps $\kappa$ is even smaller. This is the case in many applications.
In the single-swap \kmed setting, Arya et al. \cite{ARYA} find such a set of test swaps with $\alpha = 5, \beta = 1$ and $\kappa = k$ and, after scaling distances to clear denominators,
we could pick $\Delta = n \cdot \max_{i,j} \delta(i, j)$.

Consider the following local search algorithm for this generic setting.
\begin{algorithm*}[h]
 \caption{A Generic Local Search Algorithm} \label{alg:generic}
\begin{algorithmic}
\State let $\loc^0$ be any set in $\mathfrak{F}$
\State $i \leftarrow 0$
\For{$K := \lceil \kappa \cdot \ln(\Delta) \cdot M / \beta \rceil$ iterations}
\State let $\loc^{i+1}$ be the cheapest set in $\mathcal N(\loc^i)$ \Comment{could be $\loc^{i+1} = \loc^i$}
\State $i \leftarrow i+1$
\EndFor
\State \Return $\loc$
\end{algorithmic}
\end{algorithm*}
Note we do not really need to track the index $i$, we could just update the current set with the best one in its neighbourhood. The indices will be helpful in the proof.

Clearly Algorithm \ref{alg:generic} runs in polynomial time under our assumptions.
We show the approximation guarantee is what is guaranteed by local optimum solutions, even though the returned solution itself might not be a true local optimum
(i.e. it might still be that $\cost(\loc) \neq \min_{\loc' \in \mathfrak{G}(\loc)} \cost(\loc')$).
\begin{theorem}
The returned solution $\loc^K$ satisfies $\cost(\loc) \leq \frac{\alpha}{\beta} \cdot \cost(\opt)$.
\end{theorem}
\begin{proof}
We show for some $0 \leq i \leq K$ that $\beta \cdot \cost(\loc^i) - \alpha \cdot \cost(\opt) \leq 0$. As $\cost(\loc^{j+1}) \leq \cost(\loc^j)$ at each step $0 \leq j < K$,
this would show $\cost(\loc^K) \leq \cost(\loc^i) \leq \frac{\alpha}{\beta} \cdot \cost(\opt)$.

To that end, suppose $\beta \cdot \cost(\loc^i) - \alpha \cdot \cost(\opt) > 0$ for all $0 \leq i < K$; otherwise we are done. For each such $i$, \eqref{eqn:step} and the fact that the local search algorithm takes the best improvement at each step shows
\begin{eqnarray*}
\cost(\loc^{i+1}) & \leq & \cost(\loc^i) + \min_{\loc' \in \mathfrak G(\loc^i)} \left(\cost(\loc') - \cost(\loc^i)\right) \\
& \leq & \cost(\loc^i) + \frac{1}{|\mathfrak G(\loc^i)|} \sum_{\loc' \in \mathfrak G(\loc^i)} \left(\cost(\loc') - \cost(\loc^i)\right) \\
& \leq & \cost(\loc^i) + \frac{\alpha \cdot \cost(\opt) - \beta \cdot \cost(\loc^i)}{|\mathfrak G(\loc^i)|} \\
& \leq & \cost(\loc^i) + \frac{\alpha \cdot \cost(\opt) - \beta \cdot \cost(\loc^i)}{\kappa} \\
\end{eqnarray*}
Thus,
\[ \beta \cdot \cost(\loc^{i+1}) - \alpha \cdot \cost(\opt) \leq \left(1 - \frac{\beta}{\kappa}\right) \cdot \left(\beta \cdot \cost(\loc^{i}) - \alpha \cdot \cost(\opt)\right). \]
which, inductively, shows
\begin{eqnarray*}
\beta \cdot \cost(\loc^K) - \alpha \cdot \cost(\opt) & \leq & \left(1-\frac{\beta}{\kappa}\right)^K \cdot \left(\beta \cdot \cost(\loc^0) - \alpha \cdot \cost(\opt)\right) \\
& < & e^{-\ln (\Delta M)} \cdot \Delta = 1/M.
\end{eqnarray*}
As $\beta \cdot \cost(\loc^K) - \alpha \cdot \cost(\opt)$ is an integer multiple of $1/M$, then $\beta \cdot \cost(\loc^K) - \alpha \cdot \cost(\opt) \leq 0$ as required.
\end{proof}
This analysis trivially extends to the weighted swap setting, where for each $\loc' \in \mathfrak{G}(\loc)$ we have a value $\lambda_{\loc'} \geq 0$ and
\[ \sum_{\loc' \in \mathfrak{G}(\loc)} \lambda_{\loc'} \cdot (\cost(\loc') - \cost(\loc)) \leq \alpha \cdot \cost(\opt) - \beta \cdot \cost(\loc). \]
Taking $\kappa = \sum_{\loc' \in \mathfrak{G}} \lambda_{\loc'}$ yields the same conclusion: Algorithm \ref{alg:generic} will return a solution $\loc$ with $\cost(\loc) \leq \frac{\alpha}{\beta} \cdot \cost(\opt)$.

\section{A Parsimonious $L$-Reduction That Does Not Preserve Stability}\label{app:nostable}


In this section, we demonstrate that a classic $L$-reduction that reduces an instance of \qsat to one with bounded occurrence for each variable
does not necessarily preserve stability within any constant. Apart from the obvious point that the simple reduction does not work,
we wish to impart the lesson that reductions that preserve stability are not immediately obtained by classic (parsimonious) $L$-reductions: {\em stability-preserving reductions} are a distinct concept.

In particular, we show that the classic technique of replacing each occurrence of a variable
with an expander gadget fails to preserve stability. Our presentation mirrors that in \cite{vazirani}.

First, \cite{vazirani} points out that for any $k \geq 1$ that one can efficiently construct a 14-regular multigraph $G_k = (V_k, E_k)$ with $|V_k| = k$ so that $|\delta_{E_k}(S; V-S)| \geq \min\{|S|, |V-S|\}$ for any $S \subseteq V_k$.
For each variable $x \in X$ of a \qsat instance $\Phi$, let $k_x$ denote the number of clauses of $\Phi$ that depend on $x$.
Replace $x$ with $k_x$ variables in a one-to-one fashion in these clauses, call these new variables $x^1, \ldots, x^{k_x}$. Finally, for each edge $(i,j) \in G_{k_x}$ (viewing $V_{k_x}$ as integers from 1 to $k_x$) add
constraints $x^i \vee \overline{x^j}$ and $\overline{x^i} \vee x^j$. These two constraints ensure $x^i$ and $x^j$ have the same truth value.

As shown in \cite{vazirani}, this is an $L$-reduction and each variable appears in at most $29$ clauses in the resulting \qsat instance. It is also easy to verify it is a parsimonious reduction, noting any satisfying assignment
requires all copies of a variable of $\Phi$ to have the same truth value.

We demonstrate stable instances of \qsat where each variable does not appear in a bounded number of
clauses such that applying this reduction does not result in a stable instance of \qsatb.
For an integer $n \geq 1$, let $\Phi_n$ be the \sat instance with variables $X_n = \{z, x_1, x_2, \ldots, x_n\}$ and the following clauses:
\begin{itemize}
\item $z \vee \overline {x_i}$ for each $1 \leq i \leq n$,
\item $\overline{x_i}$ for each $1 \leq i \leq n$, and
\item $\overline z$.
\end{itemize}
We note this could be ``padded'' to a \tsat instance by adding gadgets like $F(z, w, y)$ instead of the clause $\overline z$ where $F(z, w, y)$ is the collection of 7 clauses that enforce all literals to be \textsc{False} (and similarly
for the other clauses),
but we stick with this smaller instance for ease of discussion.

\begin{claim}
For any $n \geq 1$, $\Phi_n$ is $\frac{1}{2}$-stable.
\end{claim}
\begin{proof}
Setting all variables in $X_n$ to \textsc{False} satisfies all clauses, call this assignment $(z^*, x^*)$. Consider some assignment $(z', x')$ and let $h = \HW{(z^*, x^*)}{(z', x')}$.

If $b$ is the number of variables that are set to \texttt{True}, then at least $b$ clauses are not satisfied, namely the singleton clauses. The fraction of unsatisfied clauses is at least $b/(2n+1)$ and
$h = b/(n+1)$. So the number of unsatisfied clauses is at least $\frac{b}{2n+1} \geq \frac{b}{2 \cdot (n+1)} = \frac{h}{2}$.
\end{proof}

Now, for $n \geq 1$ let $\Psi_n$ be the \qsatb instance for $B = 29$ that results by applying the above reduction to $\Phi_n$. Let $s_n$ denote the stability of $\Psi_n$.
\begin{claim}
$s_n \rightarrow 0$ as $n \rightarrow \infty$
\end{claim}
\begin{proof}
Note that the unique satisfying assignment for $\Psi_n$ is to set all variables to \textsc{False}. Consider the truth assignment that assigns all copies of $z$ the value \texttt{True} and all copies of each $x_i$ the value
\texttt{False}. The only clause that is not satisfied is the single clause $\overline{z^{i'}}$, for whatever copy $i'$ of $z$ was used in the singleton clause $\overline{z}$ of $\Phi_n$. So the fraction of unsatisfied clauses is $O(1/n)$.

On the other hand, number of variables in $\Psi_n$ is $3n+1$ (the total size of all clauses of $\Phi_n$) and the given assignment sets $n+1$ of them to \texttt{True}. Thus, the hamming distance between this assignment and the
all-\textsc{False} assignment is at least $1/3$. So $s_n = O(1/n)$.
\end{proof}



\end{document}